\pdfoutput=1  
\documentclass[11pt]{article}

\usepackage[T1]{fontenc}
\usepackage[utf8]{inputenc}
\usepackage{lmodern}
\usepackage{microtype}

\usepackage[margin=1in]{geometry}
\usepackage{setspace}
\usepackage{amsmath}
\usepackage{amssymb}
\usepackage{amsthm}
\usepackage{algorithm}
\usepackage{algpseudocode}

\usepackage{booktabs}
\usepackage{array}
\usepackage{multirow}
\usepackage{graphicx}
\graphicspath{{figures/}}
\usepackage{caption}
\usepackage{xcolor}

\newcolumntype{L}[1]{>{\raggedright\arraybackslash}p{#1}}

\usepackage[section]{placeins}

\usepackage{enumitem}
\usepackage[natbibapa]{apacite}
\definecolor{linkblue}{RGB}{0,68,170}
\usepackage[colorlinks=true,citecolor=linkblue,urlcolor=linkblue,linkcolor=linkblue,
            breaklinks=true]{hyperref}

\theoremstyle{definition}
\newtheorem{definition}{Definition}
\theoremstyle{plain}
\newtheorem{proposition}{Proposition}
\newtheorem{lemma}{Lemma}
\newtheorem{corollary}{Corollary}

\newcommand{\passat}[1]{\textup{pass}@#1}
\newcommand{\Tpub}{T_{\textup{pub}}}
\newcommand{\Thid}{T_{\textup{hid}}}
\newcommand{\bten}{b_{10}}
\newcommand{\bzo}{b_{01}}
\newcommand{\E}{\mathbb{E}}
\newcommand{\Prob}{\mathbb{P}}

\title{\bfseries Selection Without Signal, Recovery Through Expression:\\
A Measurement Study of Post-Hoc Falsification Operators\\
for Frozen Small Code Models}

\author{Mehmet \.{I}\c{s}can\thanks{Corresponding author. PythaLab, Y\i{}ld\i{}z
Technical University, Istanbul, Turkey. Email: \href{mailto:miscan@yildiz.edu.tr}{\texttt{miscan@yildiz.edu.tr}}.}\\[2pt]
\small PythaLab, Y\i{}ld\i{}z Technical University, Istanbul, Turkey}
\date{}

\begin{document}
\maketitle

\begin{abstract}
\noindent
Frozen small code models---at most about $1.5$ billion parameters and run locally without
fine-tuning---are attractive for offline and privacy-constrained use, but they often emit
plausible-but-wrong programs. A natural remedy is a \emph{post-hoc} operator that selects,
verifies, repairs, or otherwise re-processes the model's own samples without retraining it;
in its most principled form the operator is Popperian, attacking each candidate with a severe
test and keeping what survives. This paper measures whether such operators help. Under a single
deterministic execution oracle and a leakage-free, matched-compute protocol, twenty-six
semantic post-hoc operators---spanning selection, verification, repair, elimination,
portfolios, sound vetoes, and generation conditioning---are evaluated against a Best-of-$N$
(BoN) baseline; on the model cells and benchmarks tested, none improves held-out accuracy over
BoN. The negative is mechanistic rather than a tuning failure, and three measured forces
account for it: a \emph{coverage wall} (the weak model's hard-task failures are systematic, and
deeper sampling does not rescue them), a \emph{capability scissors} (a competent generator
leaves almost no discriminable error among the candidates that pass the visible tests), and a
\emph{near-empty consensus trap} (the visible-passing-but-hidden-wrong majority a leakage-free
selector needs almost never co-occurs with a correct alternative in the pool). A
distribution-free do-no-harm bound cannot guarantee a population harm rate of $\le\alpha$ even
at zero observed harm unless $n\ge 45$. Two operators help on a different axis, outside the
semantic output space. An expression-layer recovery (M1)---the only accuracy gain in this
study---returns correct programs the standard extractor discards through robust extraction and
public-test signature alignment; it does no harm ($\bten{=}0$), is leakage-free, and lifts
DeepSeek-Coder-1.3B by $+12$ tasks on HumanEval+ ($p{=}2.4\times10^{-4}$). An adaptive
consensus early-stop (ACE) is a calibrated compute-saving control, not a new adaptive-compute
method: $\approx 19\%$ saving at a zero-harm operating point (Hoeffding--Bentkus bound, not
corrected for $\tau$ selection), with larger savings forgoing it. The M1 gain and the selection
negative both replicate on HumanEval+ and MBPP+ across three model cells (ACE is reported on the
local benchmark only), with M1 stronger for the weaker model. The practical lesson is specific:
for a frozen small code model, fix the harness (extraction, serving, and the compute schedule)
and measure coverage before attributing failures to semantic post-hoc reasoning over the
existing samples.
\end{abstract}


\section{Introduction}
\label{sec:intro}

Small language models are an increasingly practical target for local code
generation. Models in the $0.5$--$1.5$\,B parameter range run without
quantization on a single consumer GPU, respond in seconds, and never leave the
machine, which makes them attractive for offline assistants, classroom tooling,
and privacy-constrained deployments. Their competence is measured by the
$\passat{k}$ protocol introduced with Codex and HumanEval
\citep{chen2021humaneval} and extended by MBPP \citep{austin2021mbpp}, and
code-specialized small models such as DeepSeek-Coder \citep{guo2024deepseekcoder}
and Qwen2.5-Coder \citep{hui2024qwen25coder} now bring real synthesis to this
range. Their weakness is equally well known: without large-scale fine-tuning they
fail on edge cases, misread underspecified prompts, and emit plausible-looking
but subtly wrong programs. Best-of-$N$ (BoN) sampling---draw $k$ candidates,
return the first that passes the visible tests---raises pass rates but saturates
quickly.

A tempting response is to wrap the frozen model in a smarter \emph{post-hoc}
operator: a procedure that re-processes the model's own samples, without changing
a single weight, to recover a correct answer. Four families have been tried, and
each has a documented failure mode at small scale. \emph{Selection} picks a
better sample after generation: self-consistency marginalizes over reasoning
paths \citep{wang2022selfconsistency}, and its code analogue is execution
agreement---MBR-Exec ranks by the mutual agreement of execution behaviour
\citep{shi2022mbrexec}, CodeT scores candidates against model-generated tests
\citep{chen2022codet}, and Coder-Reviewer reranks by a generation-likelihood
criterion \citep{zhang2023coderreviewer}. \emph{Verification} tries to predict
correctness directly: tool-interactive critique \citep{gou2024critic} and process
reward models \citep{lightman2024letsverify} add an external check, yet large language model
(LLM) verification of code against a specification is systematically unreliable
\citep{jin2025selfcritiquefail}, and the binding bottleneck is fault-revealing
test generation, not synthesis \citep{bansal2026vibepass}. \emph{Repair} rewrites
a failed attempt: Self-Debug \citep{chen2024selfdebug}, Self-Refine
\citep{madaan2023selfrefine}, and Reflexion \citep{shinn2023reflexion} feed
execution traces or self-critique back into a revision---but intrinsic
self-correction does not reliably improve outputs without an external signal
\citep{huang2024cannot,valmeekam2023self}, can degrade them
\citep{zhang2025darkside}, and the self-generated-test variant inherits a
documented oracle bias \citep{chen2025revisitdebug}. \emph{Test-time compute
allocation} spends the budget adaptively: it can beat flat Best-of-$N$ on capable
models \citep{snell2024testtime}, but on code the \emph{online}
difficulty-adaptive variant can drop \emph{below} best-of-$k$
\citep{damani2024hard}.

The most principled versions of these operators are \emph{Popperian}: rather than
confirm a candidate, they attack it with the most severe test they can construct
and keep what survives \citep{popper1959logic,popper1963conjectures}. Severity
has a precise error-statistical meaning---a claim is corroborated only to the
extent that it has passed a test it would probably have failed were it false
\citep{mayospanos2006severe,mayo2025severe}---and in software engineering the
same intuition appears as falsification-driven verification
\citep{groce2015verified}. The closest agentic operationalization, POPPER,
validates a fixed natural-language hypothesis by sequential severe tests under
strict Type-I control \citep{huang2025popper}; the closest engineering competitor
wraps each model call in a verified contract \citep{banerjee2026severa}; and the
REFUTE benchmark asks directly whether models can \emph{create} counterexamples
to incorrect solutions, finding that even strong reasoning agents succeed on
$<9\%$ of cases \citep{sinha2025refute}. The appeal of the program is real: it is
cheap (no training), grounded (execution is a trustworthy oracle), and it rests on a
well-established epistemology in which knowledge grows by refutation rather than
accumulation.

This paper reports what happened when that program was applied, repeatedly and
adversarially, to a frozen small code model, and gives a calibrated account of why it
did not improve it, together with the two places it did. It is the companion to an
earlier study \citep{iscan2026scaffold} that asked whether attaching Popperian
\emph{vocabulary} to a coding agent's prompt improves its code, and found that
the gains came from scaffold \emph{structure}, not from the Popperian content:
naming severe tests did not help beyond a labels-only control. That paper's
explicit next step was to ``turn the method on a sharper target''---a skill built
around \emph{executable} falsification, ``generating, running, and selecting
against real counterexamples rather than rubric impressions.'' The present work
is that successor: it replaces the LLM-judge with a deterministic execution
oracle and asks the sharper question directly---does any \emph{executable}
post-hoc falsification operator make a frozen small model write more correct code?
Because code benchmarks can leak into pre-training \citep{matton2024leakage}, any
self-authored benchmark is treated as suspect, and replication is run on the augmented
HumanEval+ and MBPP+ suites, whose hidden tests expose plausible-but-wrong
programs the original suites miss \citep{liu2023evalplus,liu2025testadequacy}.

Across twenty-six distinct operators spanning the standard families---selection,
verification, repair, counterexample search, version-space elimination,
portfolios, sound vetoes, and generation conditioning---\emph{none improves
held-out accuracy over BoN at matched compute}. What makes the result worth
reporting is that the negative is
\emph{mechanistic and measured}. Three forces, each quantified on cached
candidate pools, jointly explain every null:

\begin{itemize}[leftmargin=1.4em,itemsep=2pt,topsep=2pt]
\item \textbf{Coverage wall.} The weak model's hard-task failures are
  \emph{systematic}, not stochastic: deeper sampling does not rescue them. At
  $k{=}16$ on a held-out hard set, $16/30$ tasks still produce no hidden-correct candidate
  at all---reproducing a published online-code negative for difficulty-adaptive
  allocation \citep{damani2024hard}.
\item \textbf{Capability scissors.} A \emph{competent} generator leaves almost no
  \emph{discriminable} error among the candidates that already pass the visible
  tests; whatever gain is available lands in $\passat{1}$/coverage, not in the
  operator. The gap \emph{sharpens} with a stronger model and with explicit
  reasoning.
\item \textbf{Near-empty consensus trap.} The exact regime a leakage-free
  selector needs---a visible-passing majority that is hidden-wrong while a
  correct alternative sits in the pool---almost never co-occurs. On two
  purpose-built, sound-veto-\emph{capable} trap benchmarks the model emitted the
  triggering bug on $0/10$ and $2/16$ tasks, never as the consensus majority; and
  $\approx 83\%$ of the model's real bugs are invisible to any sound metamorphic
  relation.
\end{itemize}

\noindent These mechanisms have a sharp statistical corollary that is made
explicit and proved: a distribution-free finite-sample do-no-harm certificate
cannot bound a population harm rate at $\le\alpha$ even at \emph{zero} observed
harm unless the sample is at least $n\ge\lceil\log\delta/\log(1-\alpha)\rceil$,
which is $45$ at $\alpha{=}0.05,\delta{=}0.10$ (Proposition~\ref{prop:power}).
Every ``provable do-no-harm'' claim made on a $\le 32$-task set in this program
is therefore under-powered by construction---an observation that disciplines the
present positive results.

Two operators do beat BoN, and the reason is that they do \emph{not} operate over the
model's \emph{semantic} output space, where the three mechanisms take effect.

\begin{itemize}[leftmargin=1.4em,itemsep=2pt,topsep=2pt]
\item \textbf{Expression-layer recovery (M1).} A large fraction of a weak
  model's apparent failures are \emph{correct programs the standard extractor
  discards}---a wrong function name (\texttt{def Add} where the test calls
  \texttt{add}), prose-wrapped code, a malformed fence, multiple definitions. A
  robust multi-strategy extractor plus \emph{signature alignment} (rename the
  single defined function to the public-test name) recovers this code. Applied
  only when the standard pipeline finds no visible-passer---so it displaces no
  standard public-passing success---M1 does no harm ($\bten{=}0$ on every cell)
  and is leakage-free (the name comes from the public tests, never the reference),
  at zero additional generation. It is the only deployed accuracy gain in the
  program.
\item \textbf{Compute reduction: adaptive consensus early-stop (ACE).} BoN commits on the first
  visible-passer, so the rest of the budget is wasted; an adaptive early-stop
  recovers that compute. ACE is reported with deliberate calibration: the
  zero-harm operating point ($\bten{=}0$, with a Hoeffding--Bentkus
  Learn-then-Test calculation \citep{angelopoulos2021ltt}) saves only $\approx 19\%$ of
  samples on the competent cell; the larger $\approx 64\%$ saving sits at
  an aggressive stop that incurs measurable regressions ($\bten{=}2$) and does not
  reach the bound. The compute win is real but modest, and it collapses on the
  weak model for the same coverage-wall reason.
\end{itemize}

A measurement paper built on a self-authored benchmark faces an obvious objection:
the instrument may be the result. The M1 accuracy gain and the selection negative are
therefore replicated on two standard external benchmarks (HumanEval+ and MBPP+) across three
model cells \citep{liu2023evalplus}; ACE is reported on the local benchmark only. M1 transfers \emph{more strongly} than on the
self-authored benchmark (DeepSeek-Coder-1.3B: $29\to41$, $+12$ tasks, $p{=}2.4\times10^{-4}$ on
HumanEval+; $128\to161$, $+33$, $p{=}1.2\times10^{-10}$ on MBPP+; $\bten{=}0$
throughout; on the instruction-tuned Qwen cells ($1.5$B/$0.5$B) the gain is small
and ns, $+1$ to $+4$), and the ``no leakage-free selector beats BoN'' negative
replicates on every cell: the selector-agnostic headroom is $0$ (the consensus trap
is empty), or where it is positive ($6$--$17$ across the five remaining cells) no
realizable leakage-free selector captures it (net $\le+3$, never significant).

The contributions separate three axes and are the following.
\begin{enumerate}[leftmargin=1.6em,itemsep=2pt,topsep=2pt]
\item A \textbf{matched-compute, leakage-free measurement study} of twenty-six
  \emph{semantic output-space} operators (Definition~\ref{def:semantic}) for frozen small
  code models under a single execution oracle, with a do-no-harm bar ($\bten$, McNemar)
  fixed in advance. The headline is a \textbf{mechanistic negative} for semantic post-hoc
  accuracy: none beats BoN at matched compute on the tested cells
  (Results and Discussion).
\item Three \textbf{measured mechanisms}---coverage wall, capability scissors,
  near-empty consensus trap---that jointly explain the negative, with a proven
  finite-sample power corollary that bounds what any do-no-harm certificate can
  claim (Proposition~\ref{prop:power} and Lemma~\ref{prop:selector}, measured in the Discussion).
\item An \textbf{expression-layer recovery (M1)}, on a different axis from the twenty-six:
  a harness/extraction fix that recovers correct-but-mis-expressed code and is the only
  deployed accuracy gain in this study (do-no-harm, leakage-free; in the Results). It is
  not a semantic operator.
\item A \textbf{calibrated compute-saving control (ACE)}, not a new adaptive-compute
  method: early stopping at no measured accuracy cost, with a Hoeffding--Bentkus bound at
  the selected operating point (the zero-harm saving stated honestly against the aggressive
  saving; in the Results). It saves compute, not accuracy.
\item An \textbf{external replication} on HumanEval+ and MBPP+ of the M1 accuracy gain and
  the selection negative across three model cells, closing the self-authored-benchmark gap
  (in the Results).
\end{enumerate}

The remainder of this paper is organized as follows. Section~\ref{sec:method} presents the
post-hoc operator framework (the leakage-free, matched-compute evaluation protocol and its
formal do-no-harm test), together with the twenty-six benchmark methods under study and the
two surviving operators. Section~\ref{sec:results} reports the measurement study: the
experimental setup, the accuracy survey against BoN, the two survivors, and the external
replication on HumanEval+ and MBPP+. Section~\ref{sec:discussion} weighs the result against
prior post-hoc and falsification-based methods and locates the three measured mechanisms that
account for it. Section~\ref{sec:conclusion} concludes.

All of this is read narrowly, and nulls and negatives are reported as first-class
results. The intended takeaway is a design rule, not a no-go theorem: for a
frozen small code model the dominant addressable loss is in the harness
(extraction, serving, compute schedule), not in clever post-hoc reasoning over
the model's semantics, and accuracy claims for such operators should be interpreted
in light of the coverage wall.

\section{Method}
\label{sec:method}

A post-hoc operator sits between a frozen generator and deployment, the setup shown in
Figure~\ref{fig:framework}. The generator
draws several candidate programs for a task; the operator inspects that pool
together with only the information available at deploy time---the prompt and the
visible tests, never the hidden tests or a reference solution---and returns one
program to run. It changes no weights and, in the matched-compute regime enforced
here, draws no extra samples beyond the baseline's budget. The paper's question
is whether \emph{any} such operator---however it selects, verifies, repairs,
eliminates, or conditions generation---returns a hidden-correct program more often
than the trivial first-passer baseline. The question is fixed precisely---the pool,
the operator, and the do-no-harm test---and every method under test is then catalogued,
each with its rule and equation (\S\ref{sec:opmath}).

\begin{figure}[tbp]
\centering
\includegraphics[width=0.85\linewidth]{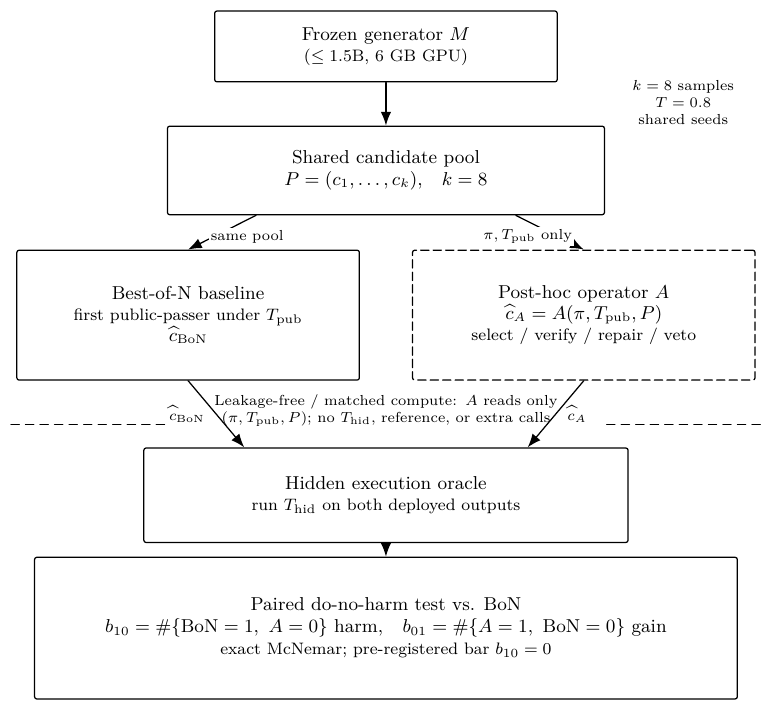}
\caption{The post-hoc operator setup. A frozen generator $M$ draws a shared candidate pool
$P=(c_1,\dots,c_k)$; a \emph{leakage-free}, \emph{matched-compute} post-hoc operator
$\mathcal{A}$ reads only the prompt and public tests $\Tpub$---never the hidden tests
$\Thid$ or a reference solution---and returns one program $\hat c$ to deploy. Every operator
is judged against the Best-of-$N$ (BoN) baseline by the do-no-harm test fixed in advance
($\bten$ harm / $\bzo$ gain, exact McNemar, bar $\bten{=}0$).}
\label{fig:framework}
\end{figure}

\noindent\textbf{Formal setup.}
A \emph{task} is a triple $t=(\pi,\Tpub,\Thid)$ of a prompt $\pi$, a set of \emph{public}
(visible) tests $\Tpub$, and a set of \emph{hidden} tests $\Thid$, each test an
(input, expected-output) pair run in a network-isolated subprocess; the pass predicate is
$\textsc{pass}(c,S)=\mathbf{1}[\,c\text{ satisfies every test in }S\,]$, so a program $c$
\emph{public-passes} when $\textsc{pass}(c,\Tpub){=}1$ and is counted \emph{correct} when
$\textsc{pass}(c,\Thid){=}1$ (on the external suites $\Tpub$ is the EvalPlus \texttt{base}
suite and $\Thid$ the augmented \texttt{plus} suite
\citep{liu2023evalplus,liu2025testadequacy}). For budget $k$ and temperature $T$ the frozen
generator $M$ produces a \emph{pool} $P=(c_1,\dots,c_k)$ with $c_i\sim M(\pi;T)$ and
per-sample seed
$\mathrm{seed}_i=\mathrm{int}\!\big(\mathrm{SHA256}(\texttt{id}\_i)[{:}8],16\big)$, so the
pool is reproducible and \emph{shared} byte-for-byte across all operators compared on it. A
\emph{post-hoc operator} is a map $\mathcal{A}\colon(\pi,\Tpub,P)\mapsto\hat c$ returning one
program $\hat c$ to deploy (an element of $P$, or a program rewritten from $P$); it is
\emph{leakage-free} when it is a function of $(\pi,\Tpub,P)$ alone, never reading $\Thid$ or a
reference solution, and \emph{matched-compute} when it issues the same number of generator
calls as the baseline (a repair or extra draw is charged against the baseline budget, so all
comparisons here are matched-compute by construction). The baseline is \emph{Best-of-$N$}
(BoN), which returns the first public-passer in pool order, $\hat c=c_{i^\star}$ with
$i^\star=\min\{i:\textsc{pass}(c_i,\Tpub){=}1\}$ (a deployed fallback if none passes), and is
itself leakage-free. Each operator is judged against BoN on $n$ tasks by the paired
off-diagonal counts \citep{fagerland2013mcnemar}
\[
  \bten=\#\{t:\textsc{BoN passes},\ \mathcal{A}\text{ fails}\}\ (\text{harm}),
  \qquad
  \bzo=\#\{t:\mathcal{A}\text{ passes},\ \textsc{BoN fails}\}\ (\text{gain}),
\]
with \emph{net gain} $\bzo-\bten$. The \emph{do-no-harm} bar, fixed before the deciding
runs, is $\bten{=}0$; a gain (a directional hypothesis) is judged by the exact one-sided
McNemar (binomial) test, which at $\bten{=}0$ gives $p=2^{-\bzo}$ and so reaches $p<0.05$ at
$\bzo\ge5$, with a stricter net-$\ge6$ floor (the two-sided exact threshold) declared in
advance. A comparison without a hypothesized direction (for instance, a leakage-free selector
against BoN) instead uses the two-sided exact test. Matched compute is essential
here, since an unmatched ``gain'' is just one more sample.

\noindent The do-no-harm bar drives every claim, so its statistical reach must be
fixed up front. At $\bten{=}0$ the bar is an observation, not yet a population
guarantee; the next proposition bounds when it can become one.

\begin{proposition}[Finite-sample floor for distribution-free do-no-harm]
\label{prop:power}
Let harm be Bernoulli with unknown rate $r$ and suppose $0$ harm events are
observed in $n$ i.i.d.\ trials. The exact one-sided upper confidence bound at level
$\delta$ is $r\le 1-\delta^{1/n}$, so certifying $r\le\alpha$ needs
$1-\delta^{1/n}\le\alpha$, i.e.\ $n\ge\lceil\log\delta/\log(1-\alpha)\rceil$. At
$\alpha{=}0.05,\ \delta{=}0.10$ this is $n\ge45$.
\end{proposition}

\noindent\emph{Proof.} $\Prob(\text{0 harm in }n)=(1-r)^n$; the $1-\delta$ upper
bound is the largest $r$ with $(1-r)^n\ge\delta$, i.e.\ $r=1-\delta^{1/n}$. Setting
this $\le\alpha$ and solving gives the bound, and
$\lceil\log0.10/\log0.95\rceil=45$. $\square$

\begin{corollary}
\label{cor:underpowered}
Any zero-harm result on a $\le32$-task benchmark (e.g.\ the sound-veto capstone and
metamorphic selection on Substrate-W) is under-powered: even at $\bten{=}0$ a
distribution-free Hoeffding--Bentkus certificate cannot issue. Only $n\ge45$
results (the ACE compute do-no-harm bound in the Results; HumanEval+ at $n{=}164$)
can carry one, so every $\le32$-task do-no-harm claim is treated as suggestive, not
proven---a discipline applied to the present positives as much as to the operators
that are rejected.
\end{corollary}

\subsection{Benchmark methods}
\label{sec:opmath}

Every method under test is catalogued here. The twenty-six operators appear in pipeline
order---a refute-first spine (\#0), then selection (\#1--5), verification (\#6--8),
repair (\#9--14), elimination/veto (\#15--17), generation conditioning (\#18--23),
and compute allocation (\#24--25)---followed by the two survivors; each entry gives
the rule or score, its equation, and the claim it makes against BoN, and verdicts
with headline numbers are in the Results and
Appendix~\ref{app:ledger}. Several methods share one leakage-free \emph{execution
latent}: a deterministic probe bank $\Pi(t)$ of at most $24$ inputs (public-test
arguments plus type-aware mutations) on which each candidate $c$ is executed in the
sandbox to a fingerprint vector $\mathrm{fp}(c)$ over $\Pi$, each coordinate one of
$\{\texttt{repr}[{:}200],\,\texttt{ERR:}\langle type\rangle,\,\texttt{NOFUNC},\,
\texttt{NOCODE},\,\texttt{TIMEOUT}\}$, at no generator cost. Two notions of agreement
recur and must not be conflated: the per-probe \emph{plurality} (the most common
fingerprint) and the \emph{strict majority} (a value with count $\ge2$ and more than
half the clean candidates). Three operators issue a finite-sample certificate from
the Hoeffding--Bentkus (HB) Learn-then-Test bound \citep{angelopoulos2021ltt},
\begin{equation}\label{eq:hb}
  \mathrm{HB}_p(\hat r;n)=\min\Big\{1,\ e^{-n\,\mathrm{KL}(\min(\hat r,\alpha)\,\|\,\alpha)},\
  e\cdot\mathrm{BinomCDF}(\lceil n\hat r\rceil;\,n,\alpha)\Big\},
\end{equation}
where $\hat r$ is the observed harm rate; the certificate issues iff
$\mathrm{HB}_p\le\delta$, which by Proposition~\ref{prop:power} needs $n\ge45$ at
$(\alpha,\delta){=}(0.05,0.10)$. The three spectral probes share one feature map
$\phi_W$, defined with its method below. The method-wide notation and the fixed
constants are collected in Table~\ref{tab:notation}.

\begin{table}[tbp]
\centering
\caption{Method-wide notation and fixed constants. Every value appears in the text and is
transcribed here for reference; externally fixed choices carry their source.}
\label{tab:notation}
\small
\begin{tabular}{@{}lll@{}}
\toprule
Symbol / setting & Meaning & Value / source \\
\midrule
$k$ & pool size (samples per task) & $8$ \\
$T$ & sampling temperature & $0.8$ \\
$|\Pi|$ & probe-bank size (deterministic inputs) & $\le 24$ \\
$\dim\phi$ & code-embedding dimension & $1536$ \\
$\alpha$ & target harm rate (do-no-harm) & $0.05$ \\
$\delta$ & certificate confidence level & $0.10$ \\
$n_{\min}$ & finite-sample floor, Proposition~\ref{prop:power} & $45$ \\
$\tau$ & ACE commit threshold (reported operating point) & $4$ \\
$\alpha_{\textsc{lin}}$ & LinUCB exploration coefficient & $1$ \\
wavelet & basis / level / pad length & \texttt{db4} / $4$ / $64$ \\
PLURAL & programme temperatures & $\{0.6,0.6,1.0,0.2,0.6\}$ \\
certificate & Learn-then-Test bound, Eq.~\eqref{eq:hb} & \citet{angelopoulos2021ltt} \\
suites & augmented external benchmarks & \citet{liu2023evalplus} \\
\bottomrule
\end{tabular}
\end{table}

The survey is organized around one distinction, fixed here and used throughout.

\begin{definition}[Semantic output-space operator]
\label{def:semantic}
An operator is \emph{semantic output-space} if it selects, verifies, repairs, eliminates,
or conditions generation based on candidate code meaning, execution behavior, repair
traces, or learned correctness features. The twenty-six operators surveyed below
(\#0--\#25) are all of this kind. The two survivors act outside this space and are analyzed
separately: expression-layer extraction (M1) acts on \emph{how} the code is written, and
compute scheduling (ACE) acts on the sampling budget.
\end{definition}

Selection operators share one governing theory, which is stated first.

\begin{definition}[Consensus trap]
\label{def:trap}
A task is in the \emph{consensus trap} for pool $P$ if (a) the public-passing
majority of $P$ is hidden-\emph{wrong} and (b) some candidate in $P$ is
hidden-\emph{correct}.
\end{definition}

\begin{lemma}[Consensus-trap headroom under a majority-symmetry assumption]
\label{prop:selector}
Let $\mathcal{S}$ be a leakage-free selector (a function of $(\pi,\Tpub,P)$ alone) that
returns a public-passer, and assume the majority-symmetry regime in which, on tasks whose
public-passing majority is hidden-\emph{correct}, deviating from that majority does not
raise expected accuracy (off-trap deviation is, in expectation, weakly harmful). Under
this assumption $\mathcal{S}$ can net-gain over BoN only on consensus-trap tasks
(Definition~\ref{def:trap}): on a non-trap task it cannot turn a BoN failure into a
success without, on the matched population of non-trap tasks, turning some BoN success
into a failure. The expected net gain is then bounded by the trap's probability mass,
$\E[\bzo-\bten]\le\Prob[\,\text{task}\in\text{trap}\,]$.
\end{lemma}

\noindent\emph{Proof sketch.} If no candidate is hidden-correct (Definition~\ref{def:trap}(b)
fails) the task is unsolvable by any selection rule. If the public-passing majority is
hidden-correct (Definition~\ref{def:trap}(a) fails), the symmetry assumption makes any
deviation from the majority weakly harmful in expectation, while a rule that never
deviates is BoN-consensus and gains nothing. So the only tasks on which a leakage-free
selector can \emph{net} gain are trap tasks. $\square$

\noindent This is a \emph{diagnostic} lemma, not a general impossibility theorem: it holds
under the stated majority-symmetry assumption and decomposes \emph{where} a leakage-free
selector's headroom can lie, rather than proving that every such selector must fail
off-trap. The assumption is plausible for execution-clustered pools: when the public-passing
majority is hidden-correct it is the modal-correct cluster, so deviating from it selects a
minority cluster that is correct strictly less often in expectation. Its force in this study
is empirical---the consensus-trap mass is measured to be near zero (in the Discussion), so the
right-hand bound is near zero on the cells tested.

\begin{itemize}[leftmargin=1.5em,itemsep=5pt,topsep=4pt]

\item \textbf{Refute-first pipeline} (REFv3, \#0). The baseline spine: draw the pool,
re-prompt the model to refute its own candidate against $\Tpub$ and revise, then
deploy the first public-passer; a family router sends counterexample and abstention
tasks to the search and abstention arms and the remainder to refute$\to$revise.
\emph{Procedural} (no closed-form score). \emph{Claim:} executable self-refutation
lifts the pass rate over plain BoN.

\item \textbf{Execution-consensus} (MBR-Exec, \#1) \citep{shi2022mbrexec}. Among
public-passers, take the per-probe plurality fingerprint $\mathrm{maj}_j$ and select
the candidate of greatest agreement with it,
\[
  \hat\imath=\arg\max_i \mathrm{agree}_i,\qquad
  \mathrm{agree}_i=\tfrac{1}{|\Pi|}\textstyle\sum_{j}\mathbf{1}\!\left[\mathrm{fp}_i(j)=\mathrm{maj}_j\right].
\]
This is the minimum-Bayes-risk-over-execution instance: a plurality vote over
deterministic execution fingerprints, not over model-generated tests.
\emph{Claim:} $\E[\bzo-\bten]>0$, which by Lemma~\ref{prop:selector} requires
the trap to carry positive mass.

\item \textbf{Embedding-medoid} (CEMS, \#2). Embed each public-passer's code to
$x_i\in\mathbb{R}^{1536}$ and return the medoid---the candidate of least total
cosine distance to the rest:
\[
  \hat\imath=\arg\min_i\textstyle\sum_{j\ne i} d_{\cos}(x_i,x_j),\qquad
  d_{\cos}(u,v)=1-\frac{u\cdot v}{\lVert u\rVert\,\lVert v\rVert}.
\]
\emph{Claim:} the central candidate is hidden-correct more often than BoN's first
passer.

\item \textbf{Behavioural-trace rerank} (BTR, \#3). The execution twin of CEMS:
replace the embedding by the fingerprint and minimise total Hamming distance,
$\hat\imath=\arg\min_i\sum_{j\ne i}h_{ij}$ with
$h_{ij}=\sum_k\mathbf{1}[\mathrm{fp}_i(k)\ne\mathrm{fp}_j(k)]$; a severity-weighted
variant weights probe $j$ by $1-(\text{plurality count}_j/n)$ so that
discriminating probes count more. \emph{Claim:} behavioural centrality beats
first-passer selection.

\item \textbf{Verisimilitude} (VRS, \#4) \citep{niiniluoto2014progress}. Rank by a
Popperian truthlikeness over probes. With $\mathrm{Ct}_T$ and $\mathrm{Ct}_F$ the
content of the severe probes a candidate survives and fails and the count proxy
$\mathrm{Ct}(S)=|S|/(|S|+1)$,
\[
  V\!s(c)=\frac{\mathrm{Ct}_T-\mathrm{Ct}_F}{\mathrm{Ct}_T+\mathrm{Ct}_F}\in[-1,1],
\]
with $V\!s{=}0$ for a tautology and $V\!s{=}{-}1$ when a clean-consensus probe makes
the candidate crash; select $\arg\max_c V\!s(c)$. \emph{Claim:} higher truthlikeness
tracks hidden correctness.

\item \textbf{Metamorphic selection} (MR, \#5). On Substrate-W, refute any candidate
that cleanly violates a declared \emph{sound} relation from the library
$\{$determinism $f(x){=}f(x)$, idempotence $f(f(x)){=}f(x)$, sortedness, length,
multiset preservation, membership $\mathrm{out}\subseteq\mathrm{in}$, permutation-
and inverse-consistency $g(f(x)){=}x\}$; a crash or timeout counts as a
non-violation, so the filter never refutes a correct program. A construction gate
enforces soundness (a known-bad must fail $\Thid$ \emph{and} violate a relation; a
known-good must satisfy all). \emph{Claim:} sound metamorphic agreement isolates the
hidden-correct candidate.

\item Where selection trusts the pool's own agreement, verification instead learns or
measures a correctness signal directly. \textbf{Learned verifier} (LCV, \#6). A logistic regressor of hidden
correctness on the $1536$-d code embedding $\phi$,
$g(\phi)=\sigma(w^\top\phi+b)$, fit with balanced classes under strict per-task
GroupKFold (no task in both train and test); deploy
$\arg\max_{i}g(\phi_i)$ among public-passers. \emph{Claim:} the area under the ROC curve
satisfies $\mathrm{AUC}(g)>\tfrac12$ leakage-free, usefully enough to rerank.

\item \textbf{Fingerprint verifier} (XFV, \#7). The same logistic $g$ on an
$8$-dimensional execution feature $\phi=$ (severity-weighted and plain consensus
agreement, error fraction, timeout/no-function fraction, Hamming distance to the
behavioural medoid, behavioural diversity, and the mean severity on agreeing and on
disagreeing probes), standardised on train folds. \emph{Claim:} execution features
separate hidden-correct above chance, tested against a label-permutation null.

\item \textbf{Latent abstention} (LEA, \#8). Force $k$ candidates under a no-refusal
prompt, embed them, and abstain when the mean pairwise cosine dispersion $\bar d$ of
the pool exceeds a train-fitted threshold
$\tau^\star=\arg\max_\tau\sum_i\mathbf{1}[(\bar d_i>\tau)=a_i]$ ($a_i$ the abstain
target), otherwise answer with the medoid. \emph{Claim:} dispersion gates ambiguous
tasks without discarding solvable ones.

\item When no candidate is right, repair tries to make one. \textbf{Self-debug} (M2, \#9) \citep{chen2024selfdebug}. One repair round,
$c'=M\!\left(\textsc{repair}(t,c,e)\right)$, the prompt branched by severity---a
\emph{close} candidate (public-passes, hidden-fails) is asked to fix only the edge
case, a \emph{far} candidate (public-fails) to fix the basic failure---given the
executed error $e$. \emph{Procedural.} \emph{Claim (matched compute):} the repaired
candidate beats a fresh draw.

\item \textbf{Bandit repair router} (BRR, \#10). A LinUCB contextual bandit selects
one of four repair arms $\{$regenerate, error-trace, edge-case, refute-revise$\}$.
Keeping per-arm $A_a$ (init the identity $I$) and $b_a$ (init $0$) with
$\theta_a=A_a^{-1}b_a$, it pulls and updates
\[
  a^\star=\arg\max_a\ \theta_a^\top x+\alpha\sqrt{x^\top A_a^{-1}x},\qquad
  A_{a^\star}\!\mathrel{+}=xx^\top,\quad b_{a^\star}\!\mathrel{+}=r\,x,
\]
where the context $x$ is a fixed Gaussian random projection of the normalised
embedding of [prompt $+$ failing code $+$ error] with a public-pass flag and a
constant bias term appended, the reward is $r=\mathbf{1}[\text{repair
hidden-passes}]$, $\alpha{=}1$, evaluated by counterfactual replay. \emph{Claim:} the learned router beats the best fixed arm.

\item \textbf{Counterexample-guided repair} (CEGIS-R, \#11). Build an ensemble
pseudo-oracle from the public-passers; a probe $j$ is a counterexample when the
target crashes while a clean plurality agrees, or disagrees with a confident
plurality:
\[
  \mathrm{CE}_j \iff \big(\neg\,\mathrm{clean}(c_j)\wedge\mathrm{clean}(\mathrm{maj}_j)\big)
  \ \vee\ \big(c_j\ne\mathrm{maj}_j\wedge \mathrm{conf}_j\ge0.6\big),\quad
  \mathrm{conf}_j=\tfrac{\text{plurality count}_j}{\#\text{voters}}.
\]
Here $\mathrm{clean}(\cdot)$ marks a non-error fingerprint and $\mathrm{maj}_j$ the
per-probe plurality of the shared latent above. Inject ``$f(x){=}A$, want $B$'' and
regenerate (up to two rounds). \emph{Claim:} a concrete counterexample repairs
better than a generic re-ask.

\item \textbf{Severity-conditioned regeneration} (SERA, \#12). Regenerate
conditioned on a \emph{severe} found counterexample and compare against a generic
re-ask and a confirmation placebo through the per-task contrast
\[
  \tau=\overline{\mathrm{pass}}\!\left(c\mid\text{severe CE}\right)
       -\overline{\mathrm{pass}}\!\left(c\mid\text{generic re-ask}\right),
\]
the load-bearing test being
$\overline{\mathrm{pass}}(\text{severe})>\overline{\mathrm{pass}}(\text{placebo})$.
\emph{Claim:} severity itself, not merely a concrete anchor, drives the gain.

\item \textbf{Prompt-normalisation} (M3, \#13). Inject the public-test signature and
its input/output examples into the prompt and draw one sample. \emph{Procedural} (no
score). \emph{Claim:} matching the weak model's expected format lifts the
single-sample pass rate.

\item \textbf{Thinking-iteration} (TIRF, \#14). Compare a visible-thinking mode to
NoThink-BoN at \emph{equal compute}: with the measured latency ratio $R$ a thinking
depth $n$ is matched to NoThink depth $\lfloor nR\rfloor$, and each arm's pass@$k$
uses the exact finite-pool hypergeometric
\[
  \passat{k}=\frac{a}{a+b}\left(1-\binom{d}{k}\Big/\binom{n}{k}\right),
\]
with $a$/$b$ the hidden-winning/-failing public-passers, $d$ the non-passers, and
pool size $n{=}8$ here.
\emph{Claim:} thinking beats NoThink at equal compute.

\item Elimination keeps only what survives refutation. \textbf{Version-space elimination} (PAVER, \#15). Treat the public-passers as
a hypothesis space and delete the refuted ones:
$\mathrm{refuted}(c)=\mathrm{mr}(c)\vee\mathrm{crash}(c)$ (sound policy) or
additionally $\vee\,\mathrm{diff}(c)$, disagreement with the strict majority (full
policy); deploy the first survivor, else BoN. Soundness is measured by the
false-refutation rate $\mathrm{FRR}$ (the fraction of hidden-correct candidates it
eliminates), and probe severity by the per-probe entropy
$-\sum_v\frac{n_v}{N}\log_2\frac{n_v}{N}$. \emph{Claim:} elimination raises accuracy
do-no-harm.

\item \textbf{Programme portfolio} (PLURAL, \#16) \citep{lakatos1968criticism}. Run
five generation programmes $\{$BoN, schema, diversity, simple, repair$\}$ at fixed
temperatures $\{0.6,0.6,1.0,0.2,0.6\}$ and deploy by a fixed tiebreak order.
\emph{Procedural.} \emph{Claim:} the portfolio beats single-programme BoN---tested
decisively against BoN's own $\passat{20}$ coverage.

\item \textbf{Conformal sound veto} (SCRC, \#17) \citep{angelopoulos2021ltt}. Keep
the largest fingerprint cluster as an anchor and override it only if the anchor
cleanly violates a sound metamorphic relation \emph{and} a clean alternative cluster
exists; certify the harm rate against full consensus by the shared $\mathrm{HB}_p$
bound of Eq.~\eqref{eq:hb} on the loss
$L_i=\mathbf{1}[\text{anchor correct}\wedge\text{vetoed}\wedge\neg\,\text{SCRC
correct}]$. \emph{Claim:} a sound veto raises accuracy with $\bten{=}0$ and a
finite-sample certificate.

\item Conditioning acts earlier, shaping the pool the operators see. \textbf{Schema hint} (PRAXIS, \#18). Prepend a family-matched
algorithmic-schema hint to the prompt, with a wrong-family hint as the placebo and
decoding fixed so that only the hint varies. \emph{Procedural.} \emph{Claim:} the
matched hint beats both vanilla and the mismatched placebo ($\Delta\ge0.03$,
$\bten\le1$).

\item \textbf{Verified-exemplar} (MNEMON, \#19). Prepend a verified
(public-passing) same-family exemplar solution to the prompt---an unrelated-family
exemplar is the placebo and a no-code preamble the scaffold---and deploy by
public-probe consensus. \emph{Procedural.} \emph{Claim:} relevance beats scaffold:
the matched exemplar must beat both vanilla and the placebo.

\item \textbf{Decoding-diversity union} (POLYGEN, \#20). Generate under diverse
decoding profiles ($\min$-$p$, typical-$p$, hot-broad) and select by the modal
public-probe cluster, $\hat c=$ the first member of
$\arg\max_{\text{cluster}}|\text{cluster}|$, falling back to BoN when the diverse
pool yields no public-passer. \emph{Claim:} diversity widens coverage the selector
can convert, as plan-level search does for capable models \citep{wang2025plansearch}.

\item \textbf{Wavelet probe} (WAVE-RL-F, \#21). Ask whether a multi-scale wavelet
basis of a candidate's code signal predicts hidden correctness beyond the
embedding. From the level-$4$ \texttt{db4} transform
$\{A,D_1,\dots,D_4\}=\mathrm{wavedec}(s)$ (signal padded to length $64$), with
detail-band energies $E_{D_j}=\langle D_j,D_j\rangle$, relative \emph{detail}
energies $\pi_j=E_{D_j}/\sum_j E_{D_j}$, and wavelet entropy
$H=-\sum_j\pi_j\log\pi_j$, the feature map collects
\[
  \phi_W(s)=\big[\log(1{+}E_{D_j}),\ \log(1{+}E_A),\ \pi_j,\ H,\ \text{band sums}\big],
\]
the approximation $A$ entering only as the separate log-energy $\log(1{+}E_A)$. GO iff the
grouped-CV $\mathrm{AUC}(\text{emb}{+}\phi_W)-\mathrm{AUC}(\text{emb})\ge0.03$ at
$p<0.05$ among public-passers. \emph{Claim:} $\phi_W$ adds discriminative signal.

\item \textbf{Frequency probe} (FREQ-RL, \#22). The same $\phi_W$ gate over richer
machine-space channels---token type, operator family, AST node type, keyword rhythm,
and character class---each an integer-coded sequence transformed by the wavelet
basis. \emph{Claim:} richer spectral channels add signal beyond the embedding.

\item \textbf{Exec-behaviour differential} (FREQ-RL\,v2, \#23). On auto-generated
inputs build each candidate's self-behaviour channels (crash, exception, output
length and change, type) and its deviation from a leave-one-out consensus and a
known-good reference; test whether deviation predicts hidden \emph{failure}
$y=1-\textsc{pass}(c,\Thid)$ above chance and whether it survives at deployment.
\emph{Claim:} behavioural deviation both detects \emph{and} repairs hidden failure
leakage-free.

\item Compute allocation changes not the pick but where the sampling budget is spent. \textbf{Adaptive allocation} (SCARF, \#24). At a fixed total budget
$b\!\cdot\!M$, greedily give the next sample to the task of maximal marginal
pass-gain (priority a difficulty signal $\hat\phi$, an oracle, or uniform), with
per-task success the exact finite-pool hypergeometric and the invariant
$\sum_i n_i=b\!\cdot\!M$ enforced on every replay. \emph{Claim:} difficulty-aware
reallocation beats flat BoN at equal compute.

\item \textbf{Easy$\to$hard reallocation} (ACE+, \#25). Lock the consensus-decided
easy tasks at the ACE stop, free $\sum_i(k_{\mathrm{base}}-n_i)$ samples, and
reinvest them on unresolved tasks up to $k{=}16$ by difficulty
$1-m_1/|\mathrm{prefix}|$ (leftover returned to the locked tasks), at matched total
compute; this reproduces the online-allocation negative \citep{damani2024hard}.
\emph{Claim:} reallocation lifts accuracy at matched compute.

\item \textbf{Expression-layer recovery} (M1, survivor; Algorithm~\ref{alg:m1}). On
tasks with \emph{no} standard public-passer, robustly re-extract a single function
from the candidate and rename it to the public-test name (the most-called bare
identifier in the asserts) by a bare-identifier substitution that leaves attribute
access intact, then deploy if it now public-passes. It fires only off the standard
pipeline, so it cannot displace a standard public-passer ($\bten{=}0$); the name
comes from public tests (leakage-free); zero extra generation. \emph{Result:} the
program's only deployed accuracy win (in the Results).

\item \textbf{Adaptive consensus early-stop} (ACE, survivor; Algorithm~\ref{alg:ace}).
Cluster the public-passers by execution fingerprint with modal cluster sizes
$m_1\ge m_2$ and stop at the first depth $n$ that is decided---lock-stop
$m_1-m_2>k-n$, or $\tau$-stop $m_1\ge\tau\wedge m_1>m_2$---returning the modal
cluster's embedding medoid; bound the harm rate against full consensus by
$\mathrm{HB}_p\le\delta$ (Eq.~\eqref{eq:hb}) at $n\ge45$. The compute saving is $1-\bar n/k$.
\emph{Result:} a modest compute reduction with a Hoeffding--Bentkus bound at the selected
operating point (in the Results).

\end{itemize}

The two surviving operators are specified in full as Algorithms~\ref{alg:m1}
and~\ref{alg:ace}; both fire only off the standard pipeline, so each is do-no-harm
by construction, and their measured effects are reported in
the Results.

\begin{algorithm}[tbp]
\caption{M1 expression-layer recovery (do-no-harm, leakage-free).}
\label{alg:m1}
\begin{algorithmic}[1]
\Require task $t$, public tests $\Tpub$, candidate pool $P=\{c_1,\dots,c_k\}$, standard extractor $E_{\mathrm{std}}$
\State $a_i \gets E_{\mathrm{std}}(c_i)$ for all $i$ \Comment{standard extraction}
\If{some $a_i$ passes $\Tpub$} \Return BoN first public-passer over $a_1,\dots,a_k$ \Comment{standard handles it: never altered}
\EndIf
\State $\mathit{name} \gets$ most-called bare identifier in $\Tpub$ \Comment{signature from public tests only (leakage-free)}
\For{$i = 1 \dots k$}
  \State $b_i \gets \textsc{RobustExtract}(c_i)$ \Comment{any \texttt{def}, prose-stripped, fence-tolerant}
  \If{$b_i$ defines exactly one function $f \neq \mathit{name}$} $b_i \gets \textsc{Rename}(b_i, f \to \mathit{name})$ \Comment{signature alignment}
  \EndIf
  \If{$b_i$ passes $\Tpub$} \Return $b_i$ \Comment{recovered}
  \EndIf
\EndFor
\State \Return BoN fallback \Comment{unchanged; $\bten{=}0$ since this branch only runs when no $a_i$ passed}
\end{algorithmic}
\end{algorithm}

\begin{algorithm}[tbp]
\caption{ACE adaptive consensus early-stop (compute reduction).}
\label{alg:ace}
\begin{algorithmic}[1]
\Require task $t$, sampler $M$, public tests $\Tpub$, budget $k$, commit threshold $\tau$, execution fingerprint $\mathrm{fp}(\cdot)$
\For{$n = 1 \dots k$}
  \State draw and grade $c_n \gets M(t)$; let $Q_n$ be the public-passers among $c_1,\dots,c_n$ \Comment{one sample per step}
  \State cluster $Q_n$ by execution fingerprint $\mathrm{fp}$; let $m_1\ge m_2$ be the two largest cluster sizes
  \If{$Q_n\neq\varnothing$ \textbf{and} \big($m_1-m_2 > k-n$ \ \textbf{or}\ $m_1\ge\tau\wedge m_1>m_2$\big)}
    \State \Return embedding medoid of the modal cluster, with $\mathrm{compute}=n$ \Comment{lock-stop / $\tau$-stop}
  \EndIf
\EndFor
\State \Return full-budget consensus pick: modal-cluster medoid over all $k$ public-passers (BoN first-passer if none), $\mathrm{compute}=k$
\Statex \textit{Reported operating point:} the smallest $\tau$ at which the early-stop incurs no measured harm ($\bten{=}0$) relative to the full-budget consensus pick (here $\tau{=}4$), read with the Hoeffding--Bentkus Learn-then-Test calculation of Eq.~\eqref{eq:hb} at $n\ge45$.
\end{algorithmic}
\end{algorithm}

\section{Results}
\label{sec:results}

All inference runs locally through the Ollama API, and no external call is made during
evaluation. The model cells are Qwen2.5-Coder-1.5B and -0.5B \citep{hui2024qwen25coder} and
DeepSeek-Coder-1.3B \citep{guo2024deepseekcoder}, with Qwen3-0.6B/1.7B (thinking and
non-thinking) on selected probes, each served on a single $6$\,GB consumer GPU. Pools are
generated once at $k{=}8$, $T{=}0.8$, and every operator is analysed post-hoc over the same
cached pools, so cross-operator comparisons share generation noise exactly. Four benchmarks are
used (Table~\ref{tab:benchmarks}): FALSIFY-BENCH-local ($52$--$80$ tasks over $7$ families:
simple functions, adversarial debugging, root-cause-vs-symptom repair, ad-hoc rescue,
ambiguous-spec abstention, faulty-oracle doubt, counterexample generation); Substrate-W ($32$
tasks, engineered so a sound metamorphic relation \emph{could} select correctly); and the
external HumanEval+ ($164$ tasks) and MBPP+ ($370$ retained, after a determinism gate drops
$8/378$ tasks with order-nondeterministic or non-serializable canonical outputs)
\citep{liu2023evalplus}. Ground truth is established by sandboxed execution of the canonical
solutions; the EvalPlus adapter was independently falsified on six checks on both external
suites, including the decisive one that a public-passing-but-wrong candidate is caught by
$\Thid$ on every sampled task (fraction $1.000$: $25/25$ on HumanEval+, $6/6$ on MBPP+). Each
operator was specified in advance with a falsifiable claim, a placebo or fair baseline where
one exists, and a decisive confound check \emph{before} its deciding run, following the design
principles for falsifiable, replicable empirical work \citep{vranjes2024design}; the protocol
is summarized in Figure~\ref{fig:protocol}.

\begin{figure}[tbp]
\centering
\includegraphics[width=0.9\linewidth]{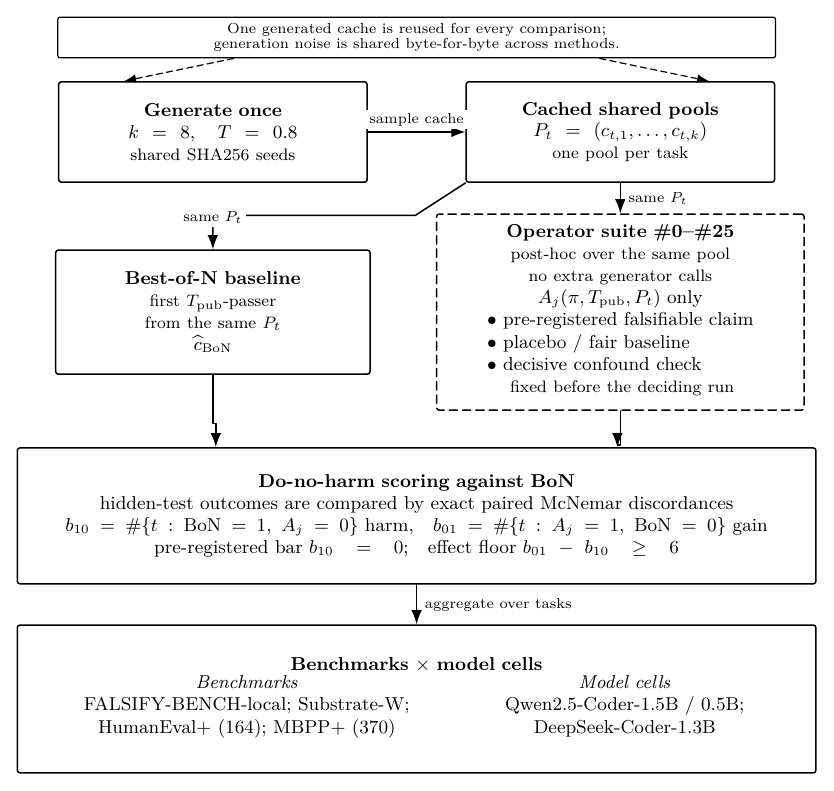}
\caption{The measurement protocol. Each task's candidate pool is generated once
($k{=}8$, $T{=}0.8$, shared SHA256 seeds) and cached; every operator is then run post-hoc over
the same pool under a falsifiable claim fixed in advance, a placebo or fair baseline, and a
decisive confound check fixed \emph{before} the deciding run, and is scored against
Best-of-$N$ (BoN) by the exact paired McNemar discordances ($\bten$ harm, $\bzo$ gain, bar
$\bten{=}0$, net $\ge6$) over the four benchmarks and the model cells.}
\label{fig:protocol}
\end{figure}

\begin{table}[tbp]
\centering
\caption{Benchmarks and test data used in the measurement study. Counts are the retained task
sizes; for MBPP+ a determinism gate drops $8/378$ tasks ($3$ order-nondeterministic, $5$
non-serializable), leaving $370$. The augmented hidden tests of HumanEval+ and MBPP+ expose
plausible-but-wrong programs the original suites miss
\citep{liu2023evalplus,liu2025testadequacy}; the model cards are
\citet{hui2024qwen25coder,guo2024deepseekcoder}.}
\label{tab:benchmarks}
\small
\begin{tabular}{@{}lllc@{}}
\toprule
Benchmark & Role & Hidden-test source & Tasks \\
\midrule
FALSIFY-BENCH-local & self-authored, $7$ families & sandboxed canonical exec & $52$--$80$ \\
Substrate-W & engineered sound-veto-capable & sandboxed canonical exec & $32$ \\
HumanEval+ & external, augmented \citep{liu2023evalplus} & EvalPlus \texttt{plus} & $164$ \\
MBPP+ & external, augmented \citep{liu2023evalplus} & EvalPlus \texttt{plus} & $370$ \\
\bottomrule
\end{tabular}
\end{table}

\phantomsection\label{sec:results-negative}%
No operator over the model's semantic output space beats Best-of-$N$ (BoN) for accuracy.
Table~\ref{tab:master} summarizes the verdicts: every operator over that space, spanning
selection, verification, repair, elimination, portfolio, sound veto, and generation
conditioning, is null, negative, or falsified against BoN at matched compute. Two patterns recur. First, operators that \emph{looked}
positive on a first pass did not survive a single confound check: the
research-programme portfolio's $+0.19$ over BoN-20 disappeared once BoN's own
$\passat{20}$ coverage was computed ($0.906$, essentially the oracle $0.938$)---the
``gain'' was BoN's broken first-passer \emph{selection}, not the portfolio.
Version-space elimination's $+2$ tasks came with a false-refutation rate of
$0.28$ (it eliminated $47/170$ hidden-\emph{correct} visible-passers, near a
coin flip), and it cleared the do-no-harm bar only through a fallback, not through
sound discrimination. Second, the one leakage-free signal that genuinely \emph{breached} the
discrimination wall---an execution-behavior differential on auto-generated probe
inputs, reaching an area under the curve (AUC) of $0.798$ for predicting hidden failure---still deployed at
$\Delta{=}0$ ($\bten{=}0$): the candidates it would re-rank already agree with
the consensus BoN commits to, so the signal has nothing to act on. Detecting a
bug and being able to \emph{do something about it without leakage} are different
problems, and only the first is solvable here.

\begin{table}[tbp]
\centering
\caption{Verdicts for the twenty-six post-hoc operators against Best-of-$N$ at
matched compute (frozen small code models). ``Do-no-harm'' is $\bten{=}0$;
``significant'' is net $\ge 6$ at $\bten{=}0$. No operator over the model's
\emph{semantics} clears both bars; the two survivors act on the \emph{expression}
and \emph{compute} axes. Full per-operator numbers
and falsification logs are in the released manifests.}
\label{tab:master}
\small
\begin{tabular}{@{}L{0.205\linewidth}L{0.36\linewidth}cL{0.20\linewidth}@{}}
\toprule
Family & Representative operators & \#ops & Verdict \\
\midrule
Selection &
  execution-consensus (MBR-Exec), embedding-medoid, behavioural rerank,
  verisimilitude, metamorphic selection &
  $5$ & null (ties BoN) \\
Verification &
  learned $\Prob(\text{hid}\mid\text{pub})$, fingerprint verifier,
  latent abstention &
  $3$ & null / below chance / net-negative gate \\
Repair &
  self-debug, bandit router, counterexample-guided, severity-conditioned,
  prompt-norm, thinking-iteration &
  $6$ & null / negative (compute-confounded or placebo-tie) \\
Elimination / portfolio / veto &
  version-space elimination, programme portfolio, conformal sound veto &
  $3$ & falsified (selection artifact / FRR$=0.28$) / $\Delta{=}0$ \\
Generation conditioning &
  schema hint, verified-exemplar, decoding-diversity, frequency/wavelet probe &
  $6$ & null / generic-scaffold / NO-GO at feature gate \\
Compute allocation &
  adaptive allocation (SCARF), easy$\to$hard reallocation (ACE+) &
  $2$ & null ($\Delta\approx0$; reproduces \citealt{damani2024hard}) \\
\midrule
Refute-first pipeline & test-aware refute$\to$revise prompt & $1$ & weak (mostly chain-of-thought, CoT; $p{=}0.23$) \\
\midrule
\textbf{Expression (survivor)} & robust extraction + signature alignment (M1) & $1$ & \textbf{positive} \\
\textbf{Compute (survivor)} & adaptive consensus early-stop (ACE) & $1$ & \textbf{modest; HB bound} \\
\bottomrule
\end{tabular}
\end{table}

A representative within-family detail: of $12$ hidden-failing candidates on the
purpose-built metamorphic benchmark, only $2$ violate any sound metamorphic
relation---so on this benchmark $\approx 83\%$ ($10/12$) of the model's hidden
failures are metamorphic-invisible, and a \emph{sound} selector (one that never
refutes a correct program) has almost nothing to fire on. Semantic-preserving
mutation exposes such inconsistencies as an auditing signal
\citep{soremekun2026mucoco}, but not, here, as a sound selector. This is not a
property of these particular relations; it is the empirical shape of where a
competent small model's errors live, and the mechanisms that account for it are
taken up in the Discussion.

\phantomsection\label{sec:survivors}%
Two operators nonetheless beat BoN, and on a \emph{different axis} than the twenty-six: an
expression-layer recovery (M1), which acts on how the model's code is written, and a modest
compute reduction (ACE) with a Hoeffding--Bentkus bound at the selected operating point, which
acts on the sampling schedule.

\phantomsection\label{sec:m1}%
The first survivor, M1, addresses a loss the capability scissors leaves untouched: that
scissors says there is no discriminable signal among the model's \emph{plausible} candidates,
whereas M1 attacks a different loss entirely, namely candidates the
model got right but \emph{expressed} in a form the standard extractor discards.
On DeepSeek-Coder-1.3B, of the candidates the standard Markdown extractor failed
to parse, a large majority contained a valid \texttt{def}; many used a wrong
function name (\texttt{def Add} where the public test calls \texttt{add}) or
wrapped the code in prose. None of this is a reasoning failure---it is an
expression/harness failure.

M1 is a robust multi-strategy extractor followed by \textbf{signature
alignment}: parse the required function name from the \emph{public tests} (the
most-called bare identifier in the assertions, never the reference solution),
locate the single defined function in the candidate, and rename it to match. It
is applied under a strict rule---\emph{recover only on tasks where the standard
pipeline finds no visible-passer}---so a task the standard pipeline already
handles is never altered. Because a deployed program is expected to pass the visible
tests, a task on which the standard pipeline yields no public-passer carries no
deployable baseline success for M1 to displace, and M1 in turn returns only a
recovered public-passer; it therefore cannot convert a standard success into a
failure, and $\bten{=}0$ is observed on every cell (Table~\ref{tab:m1}). It is
leakage-free (the name derives from public tests; an earlier dependency on the
reference solution was caught by a falsification probe and removed) and adds zero
generation (Algorithm~\ref{alg:m1}).

\begin{table}[tbp]
\centering
\caption{M1 expression-layer recovery vs.\ the standard extraction pipeline.
$\bten{=}0$ on every cell by construction (M1 fires only when the standard
pipeline finds no visible-passer). FALSIFY-BENCH-local: $n{=}52$;
HumanEval+: $n{=}164$; MBPP+: $n{=}370$. $p$ is the one-sided exact McNemar
(binomial) probability. The gain is large and significant only on the base,
low-coverage DeepSeek cell and small but do-no-harm on the instruction-tuned
Qwen cells (both $1.5$B and $0.5$B) on \emph{both} external benchmarks---an
expression-failure pattern that tracks the model, not the benchmark.}
\label{tab:m1}
\small
\begin{tabular}{@{}llcccccc@{}}
\toprule
Benchmark & Cell & std & M1 & $\bten$ & $\bzo$ & net & $p$ \\
\midrule
FALSIFY-BENCH-local & DeepSeek-Coder-1.3B & $25$ & $30$ & $0$ & $5$ & $+5$ & $0.031$ \\
FALSIFY-BENCH-local & Qwen2.5-Coder-1.5B  & $38$ & $45$ & $0$ & $7$ & $+7$ & $0.0078$ \\
\textbf{HumanEval+}  & \textbf{DeepSeek-Coder-1.3B} & $29$ & $\mathbf{41}$ & $0$ & $12$ & $\mathbf{+12}$ & $\mathbf{2.4\times10^{-4}}$ \\
\textbf{HumanEval+}  & Qwen2.5-Coder-1.5B  & $126$ & $128$ & $0$ & $2$ & $+2$ & $0.25$ (ns) \\
\textbf{HumanEval+}  & Qwen2.5-Coder-0.5B  & $114$ & $115$ & $0$ & $1$ & $+1$ & $0.5$ (ns) \\
\textbf{MBPP+}       & \textbf{DeepSeek-Coder-1.3B} & $128$ & $\mathbf{161}$ & $0$ & $33$ & $\mathbf{+33}$ & $\mathbf{1.2\times10^{-10}}$ \\
\textbf{MBPP+}       & Qwen2.5-Coder-1.5B  & $259$ & $263$ & $0$ & $4$ & $+4$ & $0.0625$ (ns) \\
\textbf{MBPP+}       & Qwen2.5-Coder-0.5B  & $199$ & $202$ & $0$ & $3$ & $+3$ & $0.125$ (ns) \\
\bottomrule
\end{tabular}
\end{table}

Table~\ref{tab:m1} reports the result. On the self-authored benchmark M1 gains $+5$
(DeepSeek) and $+7$ (Qwen) tasks, $\bten{=}0$, $p\le0.031$. The effect is
\emph{larger} on the external benchmarks for the weak model: on HumanEval+,
DeepSeek-Coder-1.3B recovers $13$ visible-passers the standard extractor dropped,
$12$ of them hidden-correct, lifting $29\to41$ ($+12$, $p{=}2.4\times10^{-4}$); on
MBPP+ it recovers $43$ dropped visible-passers, $33$ of them hidden-correct,
lifting $128\to161$ ($+33$, $p{=}1.2\times10^{-10}$). The instruction-tuned Qwen
cells (both $1.5$B and $0.5$B) gain little and ns ($+1$ to $+4$); these
instruction-tuned cells carry far fewer expression failures to recover (the $1.5$B
already public-passes $138/164$ HumanEval+ tasks). An independent
falsification harness re-derived all four checks from scratch on every cell
(every gained task had no standard visible-passer; re-extraction passes
\emph{both} public and hidden; the recovered body appears \emph{verbatim} in the
model's raw output; the aligned name is invariant to poisoning the hidden tests
and the reference): $4/4$ HOLD on all three model cells.

The honest scope: M1 \emph{recovers the model's own code}; it does not synthesize
new logic, and on the weak cell it is bounded by the coverage wall (only
$32/164$ HumanEval+ tasks yield any visible-passer at all). But within that
scope it is a real, leakage-free, do-no-harm, externally-valid accuracy win---and
the only one observed in this program. Its lesson generalizes beyond this paper: for a weak
frozen model a measurable share of ``failure'' is expression, not logic---a
surface-form dependence of the kind that moves measured performance independently
of model competence \citep{sclar2024formspread}---so the extraction harness should
be fixed before the model is blamed.

\phantomsection\label{sec:ace}%
The second survivor, ACE, is a modest compute reduction carrying a finite-sample
do-no-harm bound. A naive Best-of-$N$ draws all $k$ samples and returns the first
visible-passer, so the samples after that passer are wasted. \textbf{ACE} (adaptive
consensus early-stop) stops sampling once a commit threshold $\tau$ on agreeing
visible-passers is reached, recovering that wasted compute; its harm rate against the
full-budget consensus pick is bounded with Hoeffding--Bentkus Learn-then-Test
\citep{angelopoulos2021ltt} ($\alpha{=}0.05$, $\delta{=}0.10$; by
Proposition~\ref{prop:power} the bound can issue only at $n\ge45$).

ACE is reported with deliberate calibration, because the headline is a trade-off,
not a single number (Figure~\ref{fig:ace}, Algorithm~\ref{alg:ace}).
Table~\ref{tab:ace} gives the $\tau$-sweep on the competent cell.

\begin{figure}[tbp]
\centering
\includegraphics[width=0.66\linewidth]{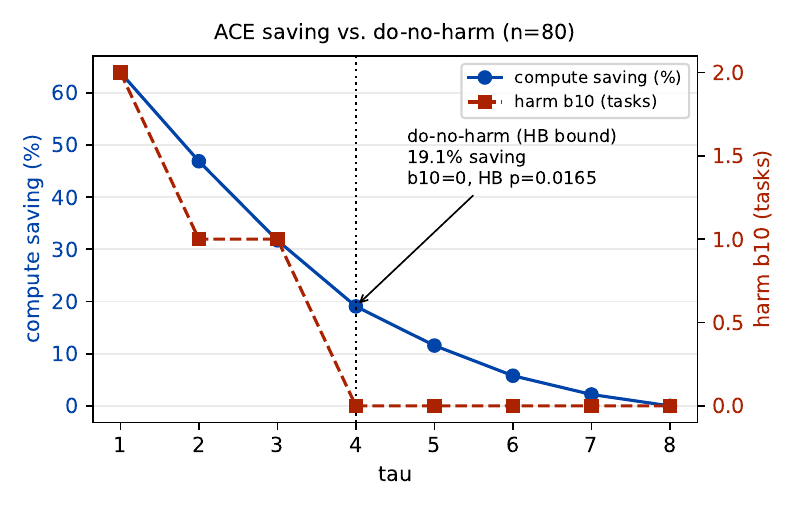}
\caption{ACE trades compute saving against do-no-harm. As the commit threshold
$\tau$ falls, the compute saving rises but so does the harm count $\bten$ (tasks the
full-budget consensus pick gets right and the early-stop gets wrong). The
zero-observed-harm operating point ($\bten{=}0$, Hoeffding--Bentkus bound $p{=}0.0165$,
not corrected for $\tau$ selection) is $\tau{=}4$ at $\approx 19\%$ saving; the
$\approx 64\%$ saving at $\tau{=}1$ comes with $\bten{=}2$ and does not reach the bound.
Regenerated from the released manifest by \texttt{figures/make\_figures.py}.}
\label{fig:ace}
\end{figure}

There is a genuine tension between \emph{saving} and \emph{do-no-harm}, and the harm
is measured against the \emph{full-budget consensus pick}---the modal-cluster medoid
taken over all $k$ public-passers, whose pass rate ($0.6125$) is within one task of
BoN's ($0.625$; $\Delta{=}0.0125$)---so that $\bten$ counts tasks where spending the
full budget would have been correct and the early stop is not. (Against a first-passer
deployment of BoN the lock-stop variant shows $\bten{=}3$, which is why the do-no-harm
reference is the full-budget consensus pick rather than the first-passer.) The aggressive
stop ($\tau{=}1$) saves $63.9\%$ of samples but
incurs $\bten{=}2$ (two such tasks) and does \emph{not} reach the bound; the
\emph{first} operating point that is empirically do-no-harm ($\bten{=}0$) is
$\tau{=}4$, which saves $19.1\%$ at Hoeffding--Bentkus $p{=}0.0165$. This minimal
harm-free $\tau$ is read off the evaluation set itself, with no separate calibration
split and no correction for scanning the $\tau$-grid, so the Hoeffding--Bentkus value
is a bound at this single operating point rather than a $\tau$-selection--corrected
guarantee. A selection-corrected certificate would require a held-out calibration split
or a multiple-$\tau$ Learn-then-Test correction, neither of which is claimed here.
The honest claim is therefore: \emph{ACE recovers $\approx 19\%$ of a full-budget
run's compute at no measured accuracy cost, with a finite-sample do-no-harm bound at
this operating point}; the larger $60\%$-class savings exist but forgo it. The
baseline is treated deliberately, because it bounds the claim: the $\approx 19\%$ is
measured against drawing every one of the $k$ samples, and a deployment that itself
stops at the first public-passer already recovers most of that compute---against such
a sequential-stop deployment the saving falls toward zero. ACE's contribution is
therefore the finite-sample do-no-harm bound on early stopping, not a large saving
over a sequential deployment.

\begin{table}[tbp]
\centering
\caption{ACE on Qwen2.5-Coder-1.5B $\times$ FALSIFY-BENCH-local ($n{=}80$; BoN
$\passat{}{=}0.625$ and the full-budget consensus pick $\passat{}{=}0.6125$ on this
cell). $\bten$ is the harm against the full-budget consensus pick (tasks it gets right
and the early stop does not); at $\tau{=}4$ the early stop equals the full-budget
consensus pick ($0.6125$, $\bten{=}\bzo{=}0$). The saving/do-no-harm trade-off is
explicit: the $\bten{=}0$ operating point is $\tau{=}4$ at $19.1\%$, not the $63.9\%$ of
$\tau{=}1$. An earlier note that bundled ``$64\%$ saving, $\bten{=}0$, certified''
conflated two $\tau$ rows; the manifest-exact values are below.}
\label{tab:ace}
\small
\begin{tabular}{@{}cccccl@{}}
\toprule
$\tau$ & ACE pass & $\bten$ & mean compute ($/8$) & saving & note \\
\midrule
$1$ & $0.625$ & $2$ & $2.89$ & $63.9\%$ & aggressive; exceeds bound ($p{=}0.53$) \\
$2$ & $0.625$ & $1$ & $4.25$ & $46.9\%$ & \\
$3$ & $0.6125$ & $1$ & $5.46$ & $31.7\%$ & \\
$\mathbf{4}$ & $0.6125$ & $\mathbf{0}$ & $6.48$ & $\mathbf{19.1\%}$ & \textbf{HB bound $p{=}0.0165$} \\
$5$--$8$ & $\approx 0.62$ & $0$ & $7.1$--$8.0$ & $11.6\%\to0\%$ & within bound, smaller \\
\bottomrule
\end{tabular}
\end{table}

The same shape holds on the second competent cell (Qwen2.5-Coder-0.5B, $n{=}50$):
the do-no-harm ($\bten{=}0$) point is $\tau{=}5$ at $2.25\%$ saving (HB $p{=}0.077$),
while the $\tau{=}1$ stop saves $61.5\%$ at $\bten{=}2$. On the $32$-task
Substrate-W cell no finite-sample bound can issue ($n<45$, Corollary~\ref{cor:underpowered});
and on DeepSeek-Coder-1.3B the saving collapses to $\sim 7$--$12\%$, the first
visible-passer arriving late under the coverage wall. Two honest framings
follow. First, ACE is a \emph{compute} win, not an accuracy win: at $k{=}8$
execution-consensus \emph{ties} BoN ($0.6125$ vs $0.625$, $\Delta{=}0.0125$), confirming
the accuracy survey above. Second, the saving is real but
modest, and it is reported as such---an instance of the same calibration applied
to the operators that are rejected.

The accuracy-facing dual of ACE fails as predicted, and is reported here for
completeness: it is not a third survivor but a
predicted null. Reinvesting ACE's freed compute as deeper sampling on hard tasks
(the easy$\to$hard move) is null at matched compute: $\Delta{=}+0.0038$
($\approx 0.3$ of $80$ tasks), the difficulty signal is inert (its allocation
equals a uniform placebo and the oracle to the digit), and the gain is not
coverage-backed. This is the
coverage wall again, and it reproduces the published online-code negative
\citep{damani2024hard}.

\phantomsection\label{sec:external}%
A measurement program built on a self-authored benchmark is only as trustworthy
as the claim that the instrument is not the result. The M1 accuracy gain and the selection
negative were therefore re-run on two standard external benchmarks---HumanEval+
($164$ tasks) and MBPP+ ($370$ tasks), each on \emph{three} model cells
(DeepSeek-Coder-1.3B, Qwen2.5-Coder-1.5B, and the smaller Qwen2.5-Coder-0.5B)
\citep{liu2023evalplus}---using the same $k{=}8$ pools and the same leakage-free,
post-hoc protocol (Figure~\ref{fig:external}). The MBPP+ ground-truth oracle is
built by executing the canonical solutions under a determinism gate that drops the
$8/378$ tasks whose canonical oracle is order-nondeterministic ($3$) or
non-serializable ($5$, e.g.\ a returned \texttt{re.Match}), so exact-match grading
is sound on every retained task; the adapter is independently falsified $6/6$ on
\emph{both} benchmarks, including the decisive discrimination-gap-transfers check
(fraction $1.000$).

\begin{figure}[tbp]
\centering
\includegraphics[width=\linewidth]{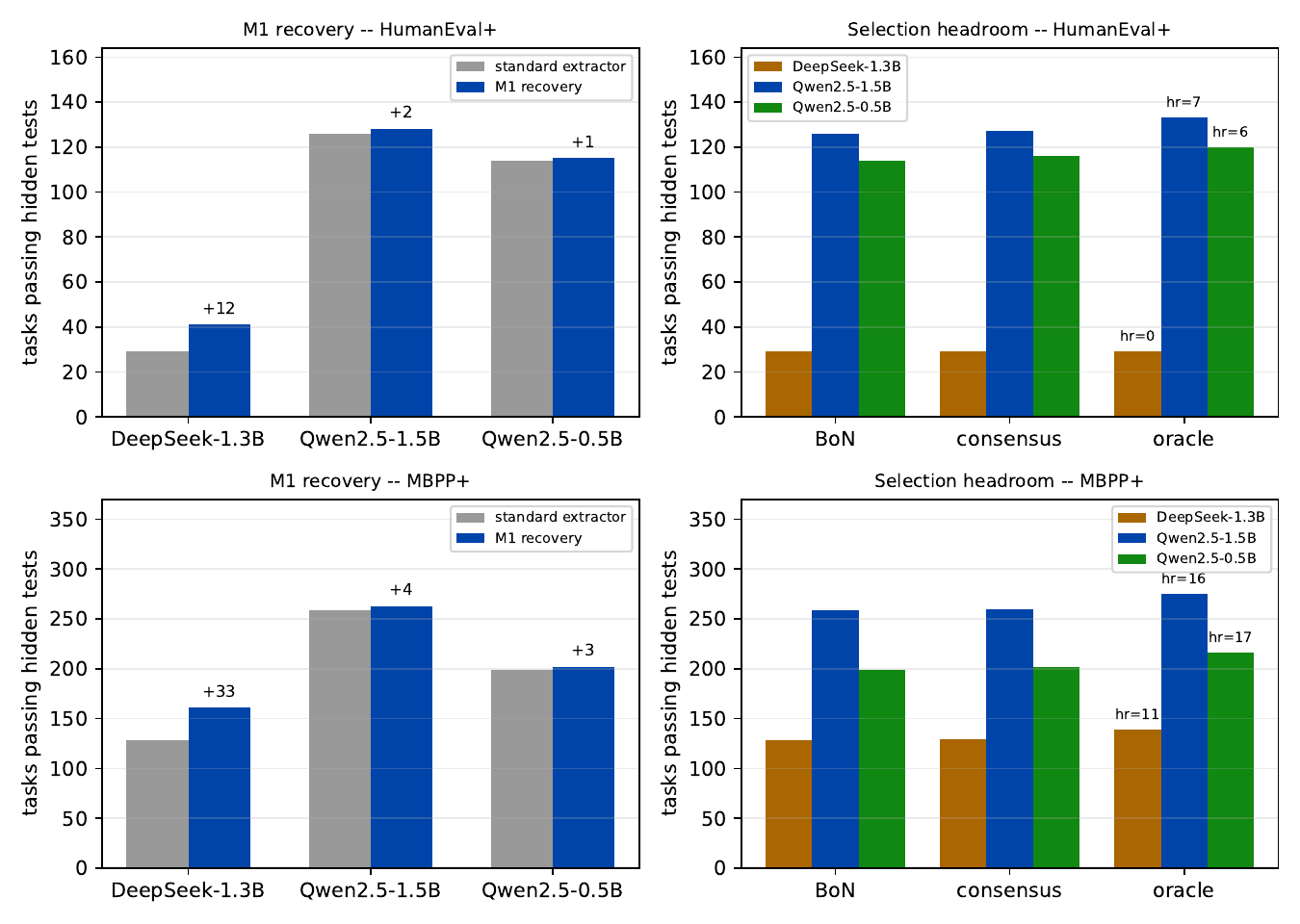}
\caption{External replication on both benchmarks and all three model cells (top:
HumanEval+, $n{=}164$; bottom: MBPP+, $n{=}370$). \emph{Left column:} M1
expression-layer recovery (standard extractor vs.\ M1) lifts the base, low-coverage
DeepSeek cell substantially ($+12$ on HumanEval+, $+33$ on MBPP+; $\bten{=}0$) and
the instruction-tuned Qwen cells ($1.5$B/$0.5$B) only slightly ($+1$ to $+4$; they
already public-pass most tasks). \emph{Right column:} the selector-agnostic ceiling
(BoN / consensus / oracle): where the oracle$-$BoN headroom is positive ($6$--$17$)
the realizable consensus selector still captures at most $+3$, never significant
($p\ge0.25$). Regenerated from the released manifests by
\texttt{figures/make\_figures.py}.}
\label{fig:external}
\end{figure}

M1 transfers, more strongly than locally, on both external benchmarks.
Table~\ref{tab:m1} gives it: DeepSeek-Coder-1.3B gains $+12$ on HumanEval+
($p{=}2.4\times10^{-4}$) and $+33$ on MBPP+ ($p{=}1.2\times10^{-10}$), both at
$\bten{=}0$; the do-no-harm and leakage-free properties survive an independent
$4/4$ falsification on every cell. The same pattern holds across both benchmarks
and all three cells: the gain is large and significant on the base, low-coverage
DeepSeek model and small but still do-no-harm on the instruction-tuned Qwen cells
($+1$ to $+4$, all ns), which already public-pass most tasks and so have few
expression failures left to recover. The expression-layer loss is therefore a
property of the weak model's output, not of the benchmark's formatting, and the
recovered gain grows with the number of tasks where the model emits
correct-but-mis-expressed code.

The selection negative replicates exactly.
Table~\ref{tab:sel} reports the selector-agnostic ceiling
$\text{headroom}=\text{oracle}-\text{BoN}$---the most a \emph{perfect} selector
among visible-passers could gain. On the weak cell, headroom is $0$: whenever
BoN's first visible-passer is hidden-wrong, \emph{no} visible-passer is
hidden-correct, so the consensus trap (Definition~\ref{def:trap}) is empty and
Lemma~\ref{prop:selector} gives a bound of $0$. Where headroom is positive it
is never leakage-free-addressable: a realizable consensus selector captures at
most net $+3$ and is never significant ($p\ge0.25$) on every such cell---HumanEval+
Qwen-1.5B/0.5B (headroom $7$/$6$), MBPP+ DeepSeek (headroom $11$), and MBPP+
Qwen-1.5B/0.5B (headroom $16$/$17$). ``No leakage-free selector beats BoN'' thus
holds on both external benchmarks and on all six cells (three models $\times$ two
benchmarks), even where a sizeable oracle headroom exists---the gap is real but not
leakage-free-addressable, a qualified strengthening of the local negative.

\begin{table}[tbp]
\centering
\caption{Selection vs.\ BoN on the external benchmarks ($k{=}8$).
$\text{headroom}=\text{oracle}-\text{BoN}$ is the selector-agnostic ceiling. Every
cell has headroom $0$ (empty consensus trap) or, where headroom exists ($6$--$17$), a
realizable consensus selector that captures at most net $+3$ of it and is never
significant ($p\ge0.25$). Here the BoN-vs-consensus $p$ is the two-sided exact
McNemar (binomial) probability, used because no gain direction is hypothesized for
selection; this contrasts with the one-sided gain test of Table~\ref{tab:m1}.}
\label{tab:sel}
\small
\begin{tabular}{@{}llccccl@{}}
\toprule
Benchmark & Cell & BoN & consensus & oracle & headroom & McNemar (two-sided) \\
\midrule
HumanEval+ & DeepSeek-Coder-1.3B & $29$ & $29$ & $29$ & $0$ & $\bten{=}0,\ \bzo{=}0$ \\
HumanEval+ & Qwen2.5-Coder-1.5B  & $126$ & $127$ & $133$ & $7$ & $\bten{=}1,\ \bzo{=}2$ ($p{=}1.0$) \\
HumanEval+ & Qwen2.5-Coder-0.5B  & $114$ & $116$ & $120$ & $6$ & $\bten{=}0,\ \bzo{=}2$ ($p{=}0.5$) \\
MBPP+ & DeepSeek-Coder-1.3B & $128$ & $129$ & $139$ & $11$ & $\bten{=}0,\ \bzo{=}1$ ($p{=}1.0$) \\
MBPP+ & Qwen2.5-Coder-1.5B  & $259$ & $260$ & $275$ & $16$ & $\bten{=}1,\ \bzo{=}2$ ($p{=}1.0$) \\
MBPP+ & Qwen2.5-Coder-0.5B  & $199$ & $202$ & $216$ & $17$ & $\bten{=}0,\ \bzo{=}3$ ($p{=}0.25$) \\
\bottomrule
\end{tabular}
\end{table}

Together, the external replication closes the gap that would otherwise sink the
study: the accuracy win (M1) and the central negative (selection cannot beat BoN)
both hold off the self-authored benchmark.

\section{Discussion}
\label{sec:discussion}

In this study, twenty-six Popperian post-hoc operators are evaluated over a frozen small code
model's own candidate pool, each judged against a matched-compute, leakage-free Best-of-$N$
(BoN) baseline by the do-no-harm (DNH) test of the formal setup ($\bten$ harm, $\bzo$ gain,
bar $\bten{=}0$). The negative, the two survivors, and the prior methods are weighed in turn,
and the mechanisms that produce them are stated with their measurements.

It is not shown that falsification is useless, nor that post-hoc operators never
help; what is shown is narrower and measurable: for a \emph{frozen small}
code model, on the benchmarks tested, no operator over the model's \emph{semantic}
output space converts to a held-out accuracy gain over matched-compute Best-of-$N$
(BoN), and the reason is structural---coverage wall, capability scissors, and a near-empty
consensus trap---not a tuning artifact. The same three mechanisms predict where
the result would change: a weaker visible-test filter (so that visible-passing,
hidden-wrong majorities become common) would refill the consensus trap, and a
genuinely \emph{coverage}-increasing intervention (a stronger generator, or
retrieval that adds new correct candidates rather than re-ranking existing ones)
would move the wall. Operators that re-rank what is already there cannot.

\phantomsection\label{sec:mechanisms}%
The observed negative is unlikely to be explained by a single tuning failure, because the
same pattern appears across operator families and confound checks. Three forces, each measured
on the cached pools, account for it, and together they bound what any do-no-harm certificate
can promise.
The first force is a coverage wall. The weak model's hard-task failures are systematic. Generating
deeper pools ($k{=}16$) for $30$ held-out hard tasks, $16/30$ still produce \emph{no
hidden-correct candidate even at} $k{=}16$; $11/30$ were already covered at $k{=}8$ (so the
loss was selection, not coverage); and only $3/30$ gain a hidden-correct candidate, of which
just $1$ gains a \emph{visible-passing} one that a leakage-free selector could grab.
Coverage---the fraction of tasks solved by \emph{any} sample, the binding axis for repeated
sampling \citep{brown2024monkeys}---does not grow here, directly reproducing the published
finding that online difficulty-adaptive allocation on code can fall below best-of-$k$
\citep{damani2024hard}: there is often nothing in the budget to reallocate \emph{to}.

The second force is a capability scissors. Conditioned on passing $\Tpub$, the residual
hidden-correctness signal is at chance. On $255$ visible-passing candidates from
Qwen2.5-Coder-1.5B a wavelet/embedding probe scores AUC $0.479$ for hidden correctness
($\Delta$AUC vs.\ embedding $-0.004$, $p{=}0.60$), even though the same features predict
\emph{public-pass} (well-formedness) at AUC $0.985$. The discriminative signal a competent
generator leaves among its plausible candidates is negligible---an instance of the
generation--verification gap that scales with pre-training compute, and so is small for a weak
model \citep{song2025mindthegap}; what gain exists has already been captured by $\passat{1}$
and by BoN's visible-test filter. The gap \emph{sharpens} with capability: enabling reasoning
on a stronger model raises $\passat{1}$ but lowers the falsification signal
(split-discriminating tasks drop from $5$ to $1$). Reaching for a second or larger model as
the judge does not escape this, since more accurate models carry \emph{more} correlated errors
\citep{kim2025correlated}. ``The model can't think'' was a $\passat{1}$ limit, not a
falsification bottleneck. The scissors has two limbs: a competent generator (Qwen) leaves no
\emph{discriminable} error among its plausible candidates, while an incompetent one (DeepSeek,
with selector-agnostic headroom $0$ on HumanEval+) leaves no \emph{recoverable} headroom at
all. The thin discriminable middle is the only place a selector could act, and a frozen small
model rarely sits in it.

The third force is the near-empty consensus trap. By Lemma~\ref{prop:selector} a leakage-free
selector can net-gain only on \emph{consensus-trap} tasks (Definition~\ref{def:trap})---those
whose public-passing majority is hidden-wrong while a correct candidate still survives in the
pool. This regime has independently been named the \emph{consensus trap} and shown to require
generator--verifier weight co-evolution to escape \citep{pan2026coverrl}; under frozen weights
its frequency is instead measured here, and it essentially never occurs. On two benchmarks
engineered to \emph{admit} a sound veto ($10$ and $16$ tasks, Qwen2.5-Coder-1.5B) the model
emitted the triggering bug on only $0/10$ and $2/16$ tasks and \emph{never} as the consensus
majority, so a conformal sound veto correctly never fired ($\Delta{=}0$, false-refute rate
$0$). A consensus-majority trap occurred in $0$ of these $26$ engineered opportunities (an
exact one-sided $95\%$ upper bound of $\approx 11\%$ on its rate), so by
Lemma~\ref{prop:selector} any leakage-free selector's expected gain is bounded at
essentially zero on the regime tested. This is the mechanistic core: the measured regime
leaves little leakage-free selector headroom---the regime in which these operators could help
is not observed to occur on the cells tested.

Finally, the certificates are power-limited. The trap's near-emptiness pushes every honest method
toward a do-no-harm claim ($\bten{=}0$), and Proposition~\ref{prop:power} fixes what that can
buy: on the $\le32$-task benchmarks where several operators were decided, no distribution-free
certificate issues even at zero observed harm (Corollary~\ref{cor:underpowered}). A formal
certificate is issued only at $n\ge45$, and every smaller do-no-harm result is read as
suggestive---a bound that disciplines the present positives as much as the operators that are
rejected.

Popper's lesson here is not that refutation fails but that \emph{refutation needs
something refutable}. A severe test is only informative when a wrong-but-plausible
hypothesis is on the table; the capability scissors says a competent generator
rarely puts one there, and the near-empty trap says the case where it does and a
correct alternative coexists is vanishingly rare. The REFUTE benchmark's finding
that strong agents create valid counterexamples for $<9\%$ of incorrect solutions
\citep{sinha2025refute} is the same wall seen from the generation side. The
companion study found that \emph{naming} severe tests adds nothing beyond
scaffold \citep{iscan2026scaffold}; this study finds that \emph{executing} them as
a post-hoc operator adds nothing to accuracy either---and locates the reason in
the data rather than the rhetoric.

Table~\ref{tab:compare} places the program against the closest prior methods, and
two distinctions are load-bearing. First, almost no prior selection or repair
method states a \emph{matched-compute, leakage-free} do-no-harm guarantee: they
report mean pass-rate gains, often on larger models and against baselines that draw
fewer samples---exactly the ``a gain is one more sample'' confound that the matched-compute
do-no-harm test rules out---and the regime their selection signal needs,
the consensus trap, is measured to be near-empty for a frozen small model
(Lemma~\ref{prop:selector}). Second, the falsification-based competitors
operate on a \emph{different object}: POPPER \citep{huang2025popper} validates a
fixed natural-language hypothesis under strict Type-I control, SEVerA
\citep{banerjee2026severa} wraps each model call in a verified contract, and REFUTE
\citep{sinha2025refute} measures whether a model can \emph{generate} a
counterexample, whereas falsification is tested here as a post-hoc operator over a frozen
model's own code, and why it does not convert to accuracy is measured. Where prior
methods have the edge is external validity and scale---they are demonstrated on
stronger models, where coverage is not the binding wall; the distinctive
contribution here is the opposite end, a calibrated account of the
small-frozen regime with the only deployed accuracy win on it.

\begin{table}[tbp]
\centering
\caption{This work among the closest post-hoc and falsification-based methods.
``DNH/cert.'' marks whether the method states a matched-compute, leakage-free
do-no-harm guarantee and a finite-sample certificate. Prior selection and repair
methods report mean gains, typically on larger models and against unmatched
baselines; the falsification competitors act on a different object than a frozen
model's own code samples.}
\label{tab:compare}
\footnotesize
\begin{tabular}{@{}L{0.235\linewidth}L{0.17\linewidth}L{0.10\linewidth}L{0.37\linewidth}@{}}
\toprule
Method & Acts on & DNH/cert.\ & Outcome and relation to this work \\
\midrule
Execution selection \citep{shi2022mbrexec,chen2022codet} & samples, by execution
  agreement & no / no & gains on capable models; null here, since the signal is in
  $\Thid$ and the consensus trap is near-empty (Lemma~\ref{prop:selector}) \\
Self-repair \citep{chen2024selfdebug,shinn2023reflexion} & a failed attempt &
  no / no & helps strong models; compute-confounded with no matched-compute gain at
  $\le1.5$\,B \citep{huang2024cannot} \\
Online allocation \citep{damani2024hard} & sampling budget & no / no & can fall
  \emph{below} best-of-$k$ on code; reproduced here \\
POPPER \citep{huang2025popper} & a fixed NL hypothesis & Type-I / -- & severe tests
  with error control; validates a hypothesis, not code synthesis \\
SEVerA \citep{banerjee2026severa} & each model call & contract / -- & verification,
  not falsification; a complementary guarantee \\
REFUTE \citep{sinha2025refute} & counterexample \emph{generation} & -- / -- &
  strong agents succeed on $<9\%$; the same wall from the generation side \\
\textbf{This work} (M1, ACE) & a frozen model's own samples & \textbf{yes / bound} &
  only deployed accuracy win on frozen small models ($+12$/$+33$, $\bten{=}0$) and a
  do-no-harm $\approx19\%$ compute saving ($n\ge45$); bounded to the small-frozen regime \\
\bottomrule
\end{tabular}
\end{table}

The two survivors point the practitioner somewhere specific. The largest cheap
win was not a clever selector but a better \emph{extractor}: M1's expression-layer
recovery delivered the program's only accuracy gain, on both benchmarks, at zero
extra compute, by fixing the harness rather than the model. The second win was a
\emph{compute} schedule (ACE), modest and with a Hoeffding--Bentkus bound at the
selected operating point. The unifying rule: spend
effort on the harness (extraction, serving, compute budget), measure coverage
($\passat{k}$) before crediting any selector, and treat post-hoc accuracy claims
with the coverage wall in mind. This is a modest conclusion, but it is the
one the measurements support.

In preparing this paper, a number-reconciliation pass caught an internal overclaim:
an earlier results note bundled ACE's aggressive $\tau{=}1$ saving ($\approx64\%$)
with the $\tau{=}4$ do-no-harm property ($\bten{=}0$, Hoeffding--Bentkus (HB) $p{=}0.0165$)
as if they held at one operating point. The manifest-exact trade-off is in
Table~\ref{tab:ace}: the zero-harm saving is $\approx19\%$. The corrected numbers
are reported here and the source notes have been fixed. This is flagged deliberately,
in the spirit of the program: a falsification discipline that does not audit its
own positives is only half a discipline.

Several boundaries of the claim deserve naming, each with its direction of risk. On scale and
models, the external replication spans two benchmarks (HumanEval+ and MBPP+), each on three
model cells (DeepSeek-Coder-1.3B and Qwen2.5-Coder-1.5B/0.5B), so the model axis covers the
$0.5$--$1.5$B range across two families; a genuinely third \emph{architecture} family would
strengthen it further, but the obvious local candidate, CodeGemma-2B, is base
(non-instruction-tuned) and does not follow the instruction protocol (it emits a bare code
fence and no solution), so it was not added, to avoid confounding the comparison. A third
external benchmark also remains worthwhile: a standalone LiveCodeBench functional floor probe
shows its easy functional subset is \emph{tractable} for these models (Qwen-1.5B solves
$37.5\%$ pass@1 on the public tests), so it is not vacuous, but a faithful integration is
deferred because LiveCodeBench ships no canonical solution and uses a class-method,
compressed-private-test contract that would require invasive changes to the validated grader.
The MBPP+ retained set excludes $8/378$ tasks whose canonical oracle is order-nondeterministic
($3$) or non-serializable ($5$), so its absolute counts are over $370$ tasks, not the nominal
$378$, and several per-operator verdicts rest on $\le 32$--$80$-task benchmarks where, by
Proposition~\ref{prop:power}, a do-no-harm certificate cannot issue; these are marked as
suggestive throughout.

On benchmark fidelity, MBPP+ passes the adapter's name-gate, and the determinism gate drops the
order-nondeterministic set-derived tasks rather than deferring them; only $3/164$ HumanEval+
tasks use floating-point tolerance, and the hidden-test input cap is a runtime bound, not a
correctness one (the discrimination-gap check still hit fraction $1.000$ at the cap). On
operator coverage, twenty-six operators is a survey, not a census; in particular,
logit-/entropy-based uncertainty channels and a training-based generator change (LoRA /
soft-prompt) were not exhaustively tested, though the generation-conditioning nulls (schema,
exemplar, thinking) and the coverage wall make a large win there unlikely. The evaluation is
oracle-only: all scoring is by hidden unit tests, and no claim is made about code quality,
abstention usefulness, or solutions that pass by hard-coding beyond what the augmented tests
catch. Finally, Ollama outputs are seed-stamped but not byte-identical across
hardware/driver/version, so aggregate rates are robust, but individual task-level discordances
in a McNemar table could shift by one or two on a different machine.

\section{Conclusion}
\label{sec:conclusion}

This study set out to make a frozen small code model write more correct code by
wrapping it in a Popperian post-hoc operator; the search was deliberately broad---twenty-six
operators, claims fixed in advance, adversarial confound checks, matched compute. For accuracy,
none of the operators that act on the model's \emph{semantics} beats Best-of-$N$ (BoN),
and the reason is measurable: the weak model's hard failures are systematic (the
coverage wall), a competent generator leaves no discriminable error among its
plausible candidates (the capability scissors), and the regime a leakage-free
selector needs is essentially empty (the near-empty consensus trap), which in
turn caps what any small-sample do-no-harm certificate can promise.

Two operators do win, and on a different axis: an expression-layer recovery (M1) that gives the
program its only deployed accuracy gain by fixing the extractor rather than the
model---do-no-harm, leakage-free, and \emph{stronger} on standard external
benchmarks ($+12$ tasks $p{=}2.4\times10^{-4}$ on HumanEval+ and $+33$
$p{=}1.2\times10^{-10}$ on MBPP+, DeepSeek-Coder-1.3B)---and a modest compute reduction
(ACE) with a Hoeffding--Bentkus bound at the selected operating point. The M1 accuracy
gain and the central selection negative both replicate on HumanEval+ and MBPP+ across
three model cells, closing the self-authored-benchmark gap. The practical message is plain and useful
to anyone deploying these models locally: fix the harness before the model, and
interpret post-hoc accuracy claims in light of the coverage wall.

These conclusions are bounded, and the boundaries point to the work that remains. The
mechanisms themselves predict where the negative would change: a genuinely
coverage-increasing intervention---a stronger generator, or retrieval that adds new
correct candidates rather than re-ranking the existing pool---would move the coverage wall,
and a weaker visible-test filter would refill the consensus trap, so both deserve a direct
test. On the measurement side, a third architecture family, a third external benchmark that
ships canonical solutions, uncertainty channels (logit/entropy), and a training-based
generator change (low-rank adaptation, LoRA, or soft-prompt) would extend the survey beyond the twenty-six
post-hoc operators examined here. Each targets a named limitation of the present study
rather than a new claim about it.

The calibrated takeaway, then, is not a universal impossibility result. For frozen small
code models in the tested regime, the dominant addressable losses were not semantic post-hoc
reasoning over existing samples, but expression recovery and compute scheduling; accuracy
gains from semantic operators would likely require changing the generator, expanding
candidate coverage, or adding genuinely new information rather than re-processing the pool
already in hand.

\section*{Generative-AI Use Disclosure}
\addcontentsline{toc}{section}{Generative-AI Use Disclosure}

This study \emph{evaluates} AI code generators (the frozen small models under test).
The study's design, hypotheses, pre-registration, literature review, statistical analysis,
and claim calibration are the author's own work. A generative-AI assistant was used only to
improve the readability of the manuscript prose (copy-editing and phrasing); it did not
design the experiments, produce or select the results, choose the citations, or set any
conclusion. No AI system is listed as an author.

\section*{Funding}
\addcontentsline{toc}{section}{Funding}

This work was supported by the Scientific and Technological Research Council of
T\"urkiye (T\"UB\.ITAK) under the 1001 programme, project no.~225M316,
``A Tilt-Trirotor Vertical Take-Off and Landing Controller Enabling Task-Oriented
Transfer and Rapid Adaptation: Hardware-in-the-Loop and Real Validation of a
Meta-Learning--Based Reinforcement Learning Architecture.'' The meta-learning code
developed in that project was integrated into the test harness together with the
frozen language models studied here, and was used during the agent-design phase to
optimize the algorithms designed by the author.

\section*{Conflict of Interest}
\addcontentsline{toc}{section}{Conflict of Interest}

The author declares no conflict of interest.

\appendix

\section{The operator ledger}
\label{app:ledger}

Table~\ref{tab:ledger} lists every post-hoc operator in the program with its
verdict and headline number against Best-of-$N$ at matched compute. Legend:
\textbf{POS} positive; \textbf{WEAK} signal, not significant; \textbf{NULL}
$\approx$ BoN; \textbf{NEG} worse than BoN; \textbf{FALS} looked positive, broke
under a confound check. Numbers are from the released per-operator manifests; cells
are Qwen2.5-Coder-1.5B/0.5B and DeepSeek-Coder-1.3B unless noted.

\begin{table}[tbp]
\centering
\caption{The post-hoc operators for frozen small code models, by family. Rows \#0--\#25 are
the twenty-six operators of the accuracy survey (all null, negative, falsified, or weak);
S1--S2 are the two survivors (M1, ACE), counted apart because they win on the expression and
compute axes the twenty-six do not address.}
\label{tab:ledger}
\footnotesize
\begin{tabular}{@{}L{0.025\linewidth}L{0.235\linewidth}L{0.12\linewidth}cL{0.38\linewidth}@{}}
\toprule
\# & Operator & Family & Verdict & Headline result vs.\ BoN \\
\midrule
0 & Refute-first pipeline (REFv3) & prompt & WEAK & $60/80$ vs $55/80$, $+6.2$pp, McNemar $p{=}0.23$ (mostly CoT edge; Popperian increment $+1.2$pp) \\
\midrule
1 & Execution-consensus (MBR-Exec) & selection & NULL & ties BoN ($k{=}8$); HumanEval+ qwen $127$ vs $126$, $p{=}1.0$ \\
2 & Embedding-medoid (CEMS) & selection & NULL & $40/50$ vs $39/50$, $p{=}1.00$ \\
3 & Behavioural-trace rerank (BTR) & selection & NULL & $40/50$, ceiling-bound \\
4 & Verisimilitude rank (VRS) & selection & NULL & $39/50$, ceiling-bound \\
5 & Metamorphic selection (MR) & selection & NULL & $+0$ both cells; $\sim$83\% of bugs MR-invisible \\
\midrule
6 & Learned verifier (LCV) & verification & NULL & ROC-AUC $0.444$ (below chance) \\
7 & Fingerprint verifier (XFV) & verification & WEAK & AUC $0.60$, permutation $p{=}0.114$ (ns) \\
8 & Latent abstention (LEA) & verification & NEG & AUC $0.738$ but $24\%$ false-abstain $\to$ net-negative \\
\midrule
9 & Self-debug (M2) & repair & NULL & compute-confounded; $\bten{=}13$ vs matched BoN@8 \\
10 & Bandit repair router (BRR) & repair & NULL & LinUCB $0.379<$ best fixed arm $0.482$ \\
11 & Counterexample-guided repair (CEGIS-R) & repair & NEG & $14/20 < $ REFv3 $16/20$; matched-compute retracted \\
12 & Severity-conditioned regen (SERA) & repair & NEG & severe $0.675=$ placebo $0.675$ ($\Delta{=}0$) \\
13 & Prompt-normalization (M3) & repair & WEAK & $+3$ at $1$-sample, $\ll$ BoN@8 \\
14 & Thinking-iteration (TIRF) & repair & NEG & $\Delta{=}{-}0.015$ (qwen3-0.6b) / ${-}0.31$ (1.7b), $\bten{=}2$ \\
\midrule
15 & Version-space elimination (PAVER) & elim./veto & FALS & $\Delta{=}{+}0.0625$ but $p{=}0.5$, FRR $0.28$, deepseek $\Delta{=}0$ \\
16 & Programme portfolio (PLURAL) & portfolio & FALS & $+0.19$ $=$ selection artifact (BoN $\passat{20}$ cov $0.906\approx$ oracle) \\
17 & Conformal sound veto (SCRC) & veto & NULL & $\Delta{=}0$ on 3 substrates; veto never fires \\
\midrule
18 & Schema hint (PRAXIS) & conditioning & NEG & matched $0.781<$ vanilla $0.844$ ($\Delta{=}{-}0.0625$) \\
19 & Verified-exemplar (MNEMON) & conditioning & FALS & real $+5$ but generic scaffold (verified $\approx$ unverified $\approx$ placebo) \\
20 & Decoding-diversity union (POLYGEN) & conditioning & NULL & coverage $+2/50$ but deployed $\Delta{=}0$ ($p{=}1.0$) \\
21 & Wavelet probe (WAVE-RL-F) & conditioning & NULL & hidden-AUC $0.479$, $\Delta$AUC ${-}0.004$ ($p{=}0.60$) \\
22 & Frequency probe (FREQ-RL) & conditioning & NULL & hidden-AUC $0.354<$ chance (gate STOP) \\
23 & Exec-behavior differential (FREQ-RL\,v2) & conditioning & NULL$^\dagger$ & breaches wall (AUC $0.798$) but deployed $\Delta{=}0$ \\
24 & Adaptive allocation (SCARF) & compute & NULL & $=$ compute-saving restated; falsification not load-bearing \\
25 & Easy$\to$hard reallocation (ACE+) & compute & NULL & $\Delta{=}{+}0.0038$, signal inert (reproduces \citealt{damani2024hard}) \\
\midrule
\multicolumn{5}{@{}l}{\textbf{Survivors (orthogonal axes):}}\\
S1 & \textbf{Expression recovery (M1)} & expression & \textbf{POS} & HumanEval+ deepseek $+12$ ($p{=}2.4{\times}10^{-4}$, $\bten{=}0$) \\
S2 & \textbf{Consensus early-stop (ACE)} & compute & \textbf{POS}$^\ddagger$ & do-no-harm $\approx19\%$ saving (HB bound, $\tau$ on eval set; $p{=}0.0165$, $n{=}80$) \\
\bottomrule
\end{tabular}

\smallskip
{\footnotesize $^\dagger$ The execution-behavior differential is the one leakage-free
signal that \emph{detects} hidden failure above chance, but it has nothing to act on
at deploy time (the consensus pick already agrees), so it converts to $\Delta{=}0$.
$^\ddagger$ ACE is a \emph{compute} win, not an accuracy win, and the
do-no-harm saving is modest (in the Results).}
\end{table}

\section{A note on finite-sample power}
\label{app:power}

Proposition~\ref{prop:power} gives the distribution-free floor: at $\alpha{=}0.05,
\delta{=}0.10$, no zero-harm do-no-harm certificate can issue below $n{=}45$. The
per-operator benchmarks span $n{=}10$ (per family) to $n{=}80$
(FALSIFY-BENCH-local) to $n{=}164$ (HumanEval+). Two consequences follow and are
applied throughout. First, every operator evaluated only on a $\le 32$-task set (the
sound-veto capstone and metamorphic selection on Substrate-W) is reported as
\emph{suggestive}: even at zero observed harm the certificate does not issue, so a
``provable do-no-harm'' phrasing there would over-claim. Second, a gain at
$\bten{=}0$ is judged by the exact one-sided McNemar/binomial, which reaches
$p<0.05$ at $\bzo\ge 5$; the pre-declared effect floor of a net $\ge 6$ (the two-sided
exact threshold) is stricter still. The only operators that clear one-sided significance are
M1's accuracy wins (FALSIFY-BENCH-local qwen $\bzo{=}7$; HumanEval+ deepseek $\bzo{=}12$;
MBPP+ deepseek $\bzo{=}33$); FALSIFY-BENCH-local deepseek ($\bzo{=}5$, $p{=}0.031$) clears
one-sided significance but not the conservative net-$\ge6$ floor, and no semantic operator
reaches even $\bzo{=}5$. No accuracy claim is therefore made for any semantic operator, and
report all such results as null/negative with their exact discordant counts rather than as
ceiling-limited non-detections. The two positives survive a family-wise correction across
the full operator set: M1's external accuracy wins clear a Bonferroni bar over all
twenty-six operators (HumanEval+ deepseek $p{=}2.4\times10^{-4}$ and MBPP+ deepseek
$p{=}1.2\times10^{-10}$, both far below $0.05/26\approx0.0019$), and no semantic operator
approaches significance even uncorrected. ACE's compute result is a single-operating-point
Hoeffding--Bentkus bound \citep{angelopoulos2021ltt}, read at the $\tau$ selected on the
evaluation set with no separate $\tau$-grid multiplicity correction (see the Results).

\bibliography{references}

\begin{thebibliography}{}

\bibitem [\protect \citeauthoryear {%
Angelopoulos%
, Bates%
, Cand\`{e}s%
, Jordan%
\BCBL {}\ \BBA {} Lei%
}{%
Angelopoulos%
\ \protect \BOthers {.}}{%
{\protect \APACyear {2021}}%
}]{%
angelopoulos2021ltt}
\APACinsertmetastar {%
angelopoulos2021ltt}%
\begin{APACrefauthors}%
Angelopoulos, A\BPBI N.%
, Bates, S.%
, Cand\`{e}s, E\BPBI J.%
, Jordan, M\BPBI I.%
\BCBL {}\ \BBA {} Lei, L.%
\end{APACrefauthors}%
\unskip\
\newblock
\APACrefYearMonthDay{2021}{}{}.
\newblock
\APACrefbtitle {Learn then Test: Calibrating Predictive Algorithms to Achieve Risk Control.} {Learn then test: Calibrating predictive algorithms to achieve risk control.}
\newblock
\APAChowpublished {arXiv preprint arXiv:2110.01052}.
\newblock
\begin{APACrefURL} \url{https://arxiv.org/abs/2110.01052} \end{APACrefURL}
\PrintBackRefs{\CurrentBib}

\bibitem [\protect \citeauthoryear {%
Austin%
\ \protect \BOthers {.}}{%
Austin%
\ \protect \BOthers {.}}{%
{\protect \APACyear {2021}}%
}]{%
austin2021mbpp}
\APACinsertmetastar {%
austin2021mbpp}%
\begin{APACrefauthors}%
Austin, J.%
, Odena, A.%
, Nye, M.%
, Bosma, M.%
, Michalewski, H.%
, Dohan, D.%
, Jiang, E.%
, Cai, C.%
, Terry, M.%
, Le, Q.%
\BCBL {}\ \BBA {} Sutton, C.%
\end{APACrefauthors}%
\unskip\
\newblock
\APACrefYearMonthDay{2021}{}{}.
\newblock
\APACrefbtitle {Program Synthesis with Large Language Models.} {Program synthesis with large language models.}
\newblock
\APAChowpublished {arXiv preprint arXiv:2108.07732}.
\newblock
\begin{APACrefURL} \url{https://arxiv.org/abs/2108.07732} \end{APACrefURL}
\PrintBackRefs{\CurrentBib}

\bibitem [\protect \citeauthoryear {%
Banerjee%
\ \protect \BOthers {.}}{%
Banerjee%
\ \protect \BOthers {.}}{%
{\protect \APACyear {2026}}%
}]{%
banerjee2026severa}
\APACinsertmetastar {%
banerjee2026severa}%
\begin{APACrefauthors}%
Banerjee, D.%
, Xu, C.%
, Ie, E.%
, Zhang, M.%
, Peng, D.%
, Lin, C\BHBI C.%
\BCBL {}\ \BBA {} Singh, G.%
\end{APACrefauthors}%
\unskip\
\newblock
\APACrefYearMonthDay{2026}{}{}.
\newblock
\APACrefbtitle {{SEVerA}: Verified Synthesis of Self-Evolving Agents.} {{SEVerA}: Verified synthesis of self-evolving agents.}
\newblock
\APAChowpublished {arXiv preprint arXiv:2603.25111}.
\newblock
\begin{APACrefURL} \url{https://arxiv.org/abs/2603.25111} \end{APACrefURL}
\PrintBackRefs{\CurrentBib}

\bibitem [\protect \citeauthoryear {%
Bansal%
\ \protect \BOthers {.}}{%
Bansal%
\ \protect \BOthers {.}}{%
{\protect \APACyear {2026}}%
}]{%
bansal2026vibepass}
\APACinsertmetastar {%
bansal2026vibepass}%
\begin{APACrefauthors}%
Bansal, S.%
, Jiao, F.%
, Zhou, Y.%
, Xu, A.%
, Joty, S.%
\BCBL {}\ \BBA {} Yavuz, S.%
\end{APACrefauthors}%
\unskip\
\newblock
\APACrefYearMonthDay{2026}{}{}.
\newblock
\APACrefbtitle {{VIBEPASS}: Can Vibe Coders Really Pass the Vibe Check?} {{VIBEPASS}: Can vibe coders really pass the vibe check?}
\newblock
\APAChowpublished {arXiv preprint arXiv:2603.15921}.
\newblock
\begin{APACrefURL} \url{https://arxiv.org/abs/2603.15921} \end{APACrefURL}
\PrintBackRefs{\CurrentBib}

\bibitem [\protect \citeauthoryear {%
Brown%
\ \protect \BOthers {.}}{%
Brown%
\ \protect \BOthers {.}}{%
{\protect \APACyear {2024}}%
}]{%
brown2024monkeys}
\APACinsertmetastar {%
brown2024monkeys}%
\begin{APACrefauthors}%
Brown, B.%
, Juravsky, J.%
, Ehrlich, R.%
, Clark, R.%
, Le, Q\BPBI V.%
, R\'{e}, C.%
\BCBL {}\ \BBA {} Mirhoseini, A.%
\end{APACrefauthors}%
\unskip\
\newblock
\APACrefYearMonthDay{2024}{}{}.
\newblock
\APACrefbtitle {Large Language Monkeys: Scaling Inference Compute with Repeated Sampling.} {Large language monkeys: Scaling inference compute with repeated sampling.}
\newblock
\APAChowpublished {arXiv preprint arXiv:2407.21787}.
\newblock
\begin{APACrefURL} \url{https://arxiv.org/abs/2407.21787} \end{APACrefURL}
\PrintBackRefs{\CurrentBib}

\bibitem [\protect \citeauthoryear {%
B.~Chen%
\ \protect \BOthers {.}}{%
B.~Chen%
\ \protect \BOthers {.}}{%
{\protect \APACyear {2023}}%
}]{%
chen2022codet}
\APACinsertmetastar {%
chen2022codet}%
\begin{APACrefauthors}%
Chen, B.%
, Zhang, F.%
, Nguyen, A.%
, Zan, D.%
, Lin, Z.%
, Lou, J\BHBI G.%
\BCBL {}\ \BBA {} Chen, W.%
\end{APACrefauthors}%
\unskip\
\newblock
\APACrefYearMonthDay{2023}{}{}.
\newblock
{\BBOQ}\APACrefatitle {{CodeT}: Code Generation with Generated Tests} {{CodeT}: Code generation with generated tests}.{\BBCQ}
\newblock
\BIn{} \APACrefbtitle {International Conference on Learning Representations (ICLR).} {International conference on learning representations (iclr).}
\newblock
\begin{APACrefURL} \url{https://arxiv.org/abs/2207.10397} \end{APACrefURL}
\newblock
\APACrefnote{arXiv:2207.10397}
\PrintBackRefs{\CurrentBib}

\bibitem [\protect \citeauthoryear {%
M.~Chen%
\ \protect \BOthers {.}}{%
M.~Chen%
\ \protect \BOthers {.}}{%
{\protect \APACyear {2021}}%
}]{%
chen2021humaneval}
\APACinsertmetastar {%
chen2021humaneval}%
\begin{APACrefauthors}%
Chen, M.%
, Tworek, J.%
, Jun, H.%
, Yuan, Q.%
, de Oliveira~Pinto, H\BPBI P.%
, Kaplan, J.%
, Edwards, H.%
, Burda, Y.%
, Joseph, N.%
, Brockman, G.%
, Ray, A.%
, Puri, R.%
, Krueger, G.%
, Petrov, M.%
, Khlaaf, H.%
, Sastry, G.%
, Mishkin, P.%
, Chan, B.%
, Gray, S.%
, Ryder, N.%
, Pavlov, M.%
, Power, A.%
, Kaiser, L.%
, Bavarian, M.%
, Winter, C.%
, Tillet, P.%
, Such, F\BPBI P.%
, Cummings, D.%
, Plappert, M.%
, Chantzis, F.%
, Barnes, E.%
, Herbert-Voss, A.%
, Guss, W\BPBI H.%
, Nichol, A.%
, Paino, A.%
, Tezak, N.%
, Tang, J.%
, Babuschkin, I.%
, Balaji, S.%
, Jain, S.%
, Saunders, W.%
, Hesse, C.%
, Carr, A\BPBI N.%
, Leike, J.%
, Achiam, J.%
, Misra, V.%
, Morikawa, E.%
, Radford, A.%
, Knight, M.%
, Brundage, M.%
, Murati, M.%
, Mayer, K.%
, Welinder, P.%
, McGrew, B.%
, Amodei, D.%
, McCandlish, S.%
, Sutskever, I.%
\BCBL {}\ \BBA {} Zaremba, W.%
\end{APACrefauthors}%
\unskip\
\newblock
\APACrefYearMonthDay{2021}{}{}.
\newblock
\APACrefbtitle {Evaluating Large Language Models Trained on Code.} {Evaluating large language models trained on code.}
\newblock
\APAChowpublished {arXiv preprint arXiv:2107.03374}.
\newblock
\begin{APACrefURL} \url{https://arxiv.org/abs/2107.03374} \end{APACrefURL}
\PrintBackRefs{\CurrentBib}

\bibitem [\protect \citeauthoryear {%
X.~Chen%
, Lin%
, Sch\"{a}rli%
\BCBL {}\ \BBA {} Zhou%
}{%
X.~Chen%
\ \protect \BOthers {.}}{%
{\protect \APACyear {2024}}%
}]{%
chen2024selfdebug}
\APACinsertmetastar {%
chen2024selfdebug}%
\begin{APACrefauthors}%
Chen, X.%
, Lin, M.%
, Sch\"{a}rli, N.%
\BCBL {}\ \BBA {} Zhou, D.%
\end{APACrefauthors}%
\unskip\
\newblock
\APACrefYearMonthDay{2024}{}{}.
\newblock
{\BBOQ}\APACrefatitle {Teaching Large Language Models to Self-Debug} {Teaching large language models to self-debug}.{\BBCQ}
\newblock
\BIn{} \APACrefbtitle {International Conference on Learning Representations (ICLR).} {International conference on learning representations (iclr).}
\newblock
\begin{APACrefURL} \url{https://arxiv.org/abs/2304.05128} \end{APACrefURL}
\newblock
\APACrefnote{arXiv:2304.05128}
\PrintBackRefs{\CurrentBib}

\bibitem [\protect \citeauthoryear {%
X.~Chen%
\ \protect \BOthers {.}}{%
X.~Chen%
\ \protect \BOthers {.}}{%
{\protect \APACyear {2025}}%
}]{%
chen2025revisitdebug}
\APACinsertmetastar {%
chen2025revisitdebug}%
\begin{APACrefauthors}%
Chen, X.%
, Tao, Z.%
, Zhang, K.%
, Zhou, C.%
, Zhang, X.%
, Gu, W.%
, He, Y.%
, Zhang, M.%
, Cai, X.%
, Zhao, H.%
\BCBL {}\ \BBA {} Jin, Z.%
\end{APACrefauthors}%
\unskip\
\newblock
\APACrefYearMonthDay{2025}{}{}.
\newblock
{\BBOQ}\APACrefatitle {Revisit Self-Debugging with Self-Generated Tests for Code Generation} {Revisit self-debugging with self-generated tests for code generation}.{\BBCQ}
\newblock
\BIn{} \APACrefbtitle {Proceedings of the 63rd Annual Meeting of the Association for Computational Linguistics (Volume 1: Long Papers)} {Proceedings of the 63rd annual meeting of the association for computational linguistics (volume 1: Long papers)}\ (\BPGS\ 18003--18023).
\newblock
\begin{APACrefURL} \url{https://doi.org/10.18653/v1/2025.acl-long.881} \end{APACrefURL}
\newblock
\APACrefnote{arXiv:2501.12793}
\PrintBackRefs{\CurrentBib}

\bibitem [\protect \citeauthoryear {%
Chou%
, Lwin%
\BCBL {}\ \BBA {} Soremekun%
}{%
Chou%
\ \protect \BOthers {.}}{%
{\protect \APACyear {2026}}%
}]{%
soremekun2026mucoco}
\APACinsertmetastar {%
soremekun2026mucoco}%
\begin{APACrefauthors}%
Chou, C\BPBI J.%
, Lwin, K\BPBI T.%
\BCBL {}\ \BBA {} Soremekun, E.%
\end{APACrefauthors}%
\unskip\
\newblock
\APACrefYearMonthDay{2026}{}{}.
\newblock
\APACrefbtitle {{MUCOCO}: Automated Consistency Testing of Code {LLMs}.} {{MUCOCO}: Automated consistency testing of code {LLMs}.}
\newblock
\APAChowpublished {arXiv preprint arXiv:2604.19086}.
\newblock
\begin{APACrefURL} \url{https://arxiv.org/abs/2604.19086} \end{APACrefURL}
\PrintBackRefs{\CurrentBib}

\bibitem [\protect \citeauthoryear {%
Damani%
, Shenfeld%
, Peng%
, Bobu%
\BCBL {}\ \BBA {} Andreas%
}{%
Damani%
\ \protect \BOthers {.}}{%
{\protect \APACyear {2025}}%
}]{%
damani2024hard}
\APACinsertmetastar {%
damani2024hard}%
\begin{APACrefauthors}%
Damani, M.%
, Shenfeld, I.%
, Peng, A.%
, Bobu, A.%
\BCBL {}\ \BBA {} Andreas, J.%
\end{APACrefauthors}%
\unskip\
\newblock
\APACrefYearMonthDay{2025}{}{}.
\newblock
{\BBOQ}\APACrefatitle {Learning How Hard to Think: Input-Adaptive Allocation of {LM} Computation} {Learning how hard to think: Input-adaptive allocation of {LM} computation}.{\BBCQ}
\newblock
\BIn{} \APACrefbtitle {International Conference on Learning Representations (ICLR).} {International conference on learning representations (iclr).}
\newblock
\begin{APACrefURL} \url{https://arxiv.org/abs/2410.04707} \end{APACrefURL}
\newblock
\APACrefnote{arXiv:2410.04707}
\PrintBackRefs{\CurrentBib}

\bibitem [\protect \citeauthoryear {%
Fagerland%
, Lydersen%
\BCBL {}\ \BBA {} Laake%
}{%
Fagerland%
\ \protect \BOthers {.}}{%
{\protect \APACyear {2013}}%
}]{%
fagerland2013mcnemar}
\APACinsertmetastar {%
fagerland2013mcnemar}%
\begin{APACrefauthors}%
Fagerland, M\BPBI W.%
, Lydersen, S.%
\BCBL {}\ \BBA {} Laake, P.%
\end{APACrefauthors}%
\unskip\
\newblock
\APACrefYearMonthDay{2013}{}{}.
\newblock
{\BBOQ}\APACrefatitle {The {McNemar} Test for Binary Matched-Pairs Data: Mid-\textit{p} and Asymptotic Are Better than Exact Conditional} {The {McNemar} test for binary matched-pairs data: Mid-\textit{p} and asymptotic are better than exact conditional}.{\BBCQ}
\newblock
\APACjournalVolNumPages{BMC Medical Research Methodology}{13}{1}{91}.
\newblock
\begin{APACrefURL} \url{https://doi.org/10.1186/1471-2288-13-91} \end{APACrefURL}
\PrintBackRefs{\CurrentBib}

\bibitem [\protect \citeauthoryear {%
Gou%
\ \protect \BOthers {.}}{%
Gou%
\ \protect \BOthers {.}}{%
{\protect \APACyear {2024}}%
}]{%
gou2024critic}
\APACinsertmetastar {%
gou2024critic}%
\begin{APACrefauthors}%
Gou, Z.%
, Shao, Z.%
, Gong, Y.%
, Shen, Y.%
, Yang, Y.%
, Duan, N.%
\BCBL {}\ \BBA {} Chen, W.%
\end{APACrefauthors}%
\unskip\
\newblock
\APACrefYearMonthDay{2024}{}{}.
\newblock
{\BBOQ}\APACrefatitle {{CRITIC}: Large Language Models Can Self-Correct with Tool-Interactive Critiquing} {{CRITIC}: Large language models can self-correct with tool-interactive critiquing}.{\BBCQ}
\newblock
\BIn{} \APACrefbtitle {International Conference on Learning Representations (ICLR).} {International conference on learning representations (iclr).}
\newblock
\begin{APACrefURL} \url{https://arxiv.org/abs/2305.11738} \end{APACrefURL}
\newblock
\APACrefnote{arXiv:2305.11738}
\PrintBackRefs{\CurrentBib}

\bibitem [\protect \citeauthoryear {%
Groce%
, Ahmed%
, Jensen%
\BCBL {}\ \BBA {} McKenney%
}{%
Groce%
\ \protect \BOthers {.}}{%
{\protect \APACyear {2015}}%
}]{%
groce2015verified}
\APACinsertmetastar {%
groce2015verified}%
\begin{APACrefauthors}%
Groce, A.%
, Ahmed, I.%
, Jensen, C.%
\BCBL {}\ \BBA {} McKenney, P\BPBI E.%
\end{APACrefauthors}%
\unskip\
\newblock
\APACrefYearMonthDay{2015}{}{}.
\newblock
{\BBOQ}\APACrefatitle {How Verified is My Code? {F}alsification-Driven Verification} {How verified is my code? {F}alsification-driven verification}.{\BBCQ}
\newblock
\BIn{} \APACrefbtitle {Proceedings of the 2015 30th IEEE/ACM International Conference on Automated Software Engineering (ASE)} {Proceedings of the 2015 30th ieee/acm international conference on automated software engineering (ase)}\ (\BPGS\ 737--748).
\newblock
\begin{APACrefURL} \url{https://doi.org/10.1109/ASE.2015.40} \end{APACrefURL}
\PrintBackRefs{\CurrentBib}

\bibitem [\protect \citeauthoryear {%
Guo%
\ \protect \BOthers {.}}{%
Guo%
\ \protect \BOthers {.}}{%
{\protect \APACyear {2024}}%
}]{%
guo2024deepseekcoder}
\APACinsertmetastar {%
guo2024deepseekcoder}%
\begin{APACrefauthors}%
Guo, D.%
, Zhu, Q.%
, Yang, D.%
, Xie, Z.%
, Dong, K.%
, Zhang, W.%
, Chen, G.%
, Bi, X.%
, Wu, Y.%
, Li, Y.%
, Luo, F.%
, Xiong, Y.%
\BCBL {}\ \BBA {} Liang, W.%
\end{APACrefauthors}%
\unskip\
\newblock
\APACrefYearMonthDay{2024}{}{}.
\newblock
\APACrefbtitle {{DeepSeek-Coder}: When the Large Language Model Meets Programming -- The Rise of Code Intelligence.} {{DeepSeek-Coder}: When the large language model meets programming -- the rise of code intelligence.}
\newblock
\APAChowpublished {arXiv preprint arXiv:2401.14196}.
\newblock
\begin{APACrefURL} \url{https://arxiv.org/abs/2401.14196} \end{APACrefURL}
\PrintBackRefs{\CurrentBib}

\bibitem [\protect \citeauthoryear {%
J.~Huang%
\ \protect \BOthers {.}}{%
J.~Huang%
\ \protect \BOthers {.}}{%
{\protect \APACyear {2024}}%
}]{%
huang2024cannot}
\APACinsertmetastar {%
huang2024cannot}%
\begin{APACrefauthors}%
Huang, J.%
, Chen, X.%
, Mishra, S.%
, Zheng, H\BPBI S.%
, Yu, A\BPBI W.%
, Song, X.%
\BCBL {}\ \BBA {} Zhou, D.%
\end{APACrefauthors}%
\unskip\
\newblock
\APACrefYearMonthDay{2024}{}{}.
\newblock
{\BBOQ}\APACrefatitle {Large Language Models Cannot Self-Correct Reasoning Yet} {Large language models cannot self-correct reasoning yet}.{\BBCQ}
\newblock
\BIn{} \APACrefbtitle {International Conference on Learning Representations (ICLR).} {International conference on learning representations (iclr).}
\newblock
\begin{APACrefURL} \url{https://arxiv.org/abs/2310.01798} \end{APACrefURL}
\newblock
\APACrefnote{arXiv:2310.01798}
\PrintBackRefs{\CurrentBib}

\bibitem [\protect \citeauthoryear {%
K.~Huang%
\ \protect \BOthers {.}}{%
K.~Huang%
\ \protect \BOthers {.}}{%
{\protect \APACyear {2025}}%
}]{%
huang2025popper}
\APACinsertmetastar {%
huang2025popper}%
\begin{APACrefauthors}%
Huang, K.%
, Jin, Y.%
, Li, R.%
, Li, M\BPBI Y.%
, Cand\`{e}s, E.%
\BCBL {}\ \BBA {} Leskovec, J.%
\end{APACrefauthors}%
\unskip\
\newblock
\APACrefYearMonthDay{2025}{}{}.
\newblock
{\BBOQ}\APACrefatitle {Automated Hypothesis Validation with Agentic Sequential Falsifications} {Automated hypothesis validation with agentic sequential falsifications}.{\BBCQ}
\newblock
\BIn{} \APACrefbtitle {Proceedings of the 42nd International Conference on Machine Learning (ICML).} {Proceedings of the 42nd international conference on machine learning (icml).}
\newblock
\begin{APACrefURL} \url{https://arxiv.org/abs/2502.09858} \end{APACrefURL}
\newblock
\APACrefnote{POPPER; arXiv:2502.09858}
\PrintBackRefs{\CurrentBib}

\bibitem [\protect \citeauthoryear {%
Hui%
\ \protect \BOthers {.}}{%
Hui%
\ \protect \BOthers {.}}{%
{\protect \APACyear {2024}}%
}]{%
hui2024qwen25coder}
\APACinsertmetastar {%
hui2024qwen25coder}%
\begin{APACrefauthors}%
Hui, B.%
, Yang, J.%
, Cui, Z.%
, Yang, J.%
, Liu, D.%
, Zhang, L.%
, Liu, T.%
, Zhang, J.%
, Yu, B.%
, Lu, K.%
, Dang, K.%
, Fan, Y.%
, Zhang, Y.%
, Yang, A.%
, Men, R.%
, Huang, F.%
, Zheng, B.%
, Miao, Y.%
, Quan, S.%
, Feng, Y.%
, Ren, X.%
, Ren, X.%
, Zhou, J.%
\BCBL {}\ \BBA {} Lin, J.%
\end{APACrefauthors}%
\unskip\
\newblock
\APACrefYearMonthDay{2024}{}{}.
\newblock
\APACrefbtitle {{Qwen2.5-Coder} Technical Report.} {{Qwen2.5-Coder} technical report.}
\newblock
\APAChowpublished {arXiv preprint arXiv:2409.12186}.
\newblock
\begin{APACrefURL} \url{https://arxiv.org/abs/2409.12186} \end{APACrefURL}
\PrintBackRefs{\CurrentBib}

\bibitem [\protect \citeauthoryear {%
\.{I}\c{s}can%
}{%
\.{I}\c{s}can%
}{%
{\protect \APACyear {2026}}%
}]{%
iscan2026scaffold}
\APACinsertmetastar {%
iscan2026scaffold}%
\begin{APACrefauthors}%
\.{I}\c{s}can, M.%
\end{APACrefauthors}%
\unskip\
\newblock
\APACrefYearMonthDay{2026}{}{}.
\newblock
\APACrefbtitle {Scaffold, Not Vocabulary? {A} Controlled, Two-Tier, Pre-Registered Study of a Popperian Code-Generation Skill.} {Scaffold, not vocabulary? {A} controlled, two-tier, pre-registered study of a popperian code-generation skill.}
\newblock
\APAChowpublished {arXiv preprint arXiv:2606.06454}.
\newblock
\begin{APACrefURL} \url{https://arxiv.org/abs/2606.06454} \end{APACrefURL}
\newblock
\APACrefnote{Companion preprint; PythaLab, Y\i{}ld\i{}z Technical University}
\PrintBackRefs{\CurrentBib}

\bibitem [\protect \citeauthoryear {%
Jin%
\ \BBA {} Chen%
}{%
Jin%
\ \BBA {} Chen%
}{%
{\protect \APACyear {2025}}%
}]{%
jin2025selfcritiquefail}
\APACinsertmetastar {%
jin2025selfcritiquefail}%
\begin{APACrefauthors}%
Jin, H.%
\BCBT {}\ \BBA {} Chen, H.%
\end{APACrefauthors}%
\unskip\
\newblock
\APACrefYearMonthDay{2025}{}{}.
\newblock
{\BBOQ}\APACrefatitle {Uncovering Systematic Failures of {LLMs} in Verifying Code Against Natural Language Specifications} {Uncovering systematic failures of {LLMs} in verifying code against natural language specifications}.{\BBCQ}
\newblock
\BIn{} \APACrefbtitle {2025 40th IEEE/ACM International Conference on Automated Software Engineering (ASE)} {2025 40th ieee/acm international conference on automated software engineering (ase)}\ (\BPGS\ 3819--3823).
\newblock
\begin{APACrefURL} \url{https://doi.org/10.1109/ASE63991.2025.00323} \end{APACrefURL}
\newblock
\APACrefnote{NIER track; arXiv:2508.12358}
\PrintBackRefs{\CurrentBib}

\bibitem [\protect \citeauthoryear {%
Kim%
, Garg%
, Peng%
\BCBL {}\ \BBA {} Garg%
}{%
Kim%
\ \protect \BOthers {.}}{%
{\protect \APACyear {2025}}%
}]{%
kim2025correlated}
\APACinsertmetastar {%
kim2025correlated}%
\begin{APACrefauthors}%
Kim, E.%
, Garg, A.%
, Peng, K.%
\BCBL {}\ \BBA {} Garg, N.%
\end{APACrefauthors}%
\unskip\
\newblock
\APACrefYearMonthDay{2025}{}{}.
\newblock
{\BBOQ}\APACrefatitle {Correlated Errors in Large Language Models} {Correlated errors in large language models}.{\BBCQ}
\newblock
\BIn{} \APACrefbtitle {Proceedings of the 42nd International Conference on Machine Learning (ICML).} {Proceedings of the 42nd international conference on machine learning (icml).}
\newblock
\begin{APACrefURL} \url{https://arxiv.org/abs/2506.07962} \end{APACrefURL}
\newblock
\APACrefnote{arXiv:2506.07962}
\PrintBackRefs{\CurrentBib}

\bibitem [\protect \citeauthoryear {%
Lakatos%
}{%
Lakatos%
}{%
{\protect \APACyear {1968}}%
}]{%
lakatos1968criticism}
\APACinsertmetastar {%
lakatos1968criticism}%
\begin{APACrefauthors}%
Lakatos, I.%
\end{APACrefauthors}%
\unskip\
\newblock
\APACrefYearMonthDay{1968}{}{}.
\newblock
{\BBOQ}\APACrefatitle {Criticism and the Methodology of Scientific Research Programmes} {Criticism and the methodology of scientific research programmes}.{\BBCQ}
\newblock
\APACjournalVolNumPages{Proceedings of the Aristotelian Society}{69}{}{149--186}.
\newblock
\begin{APACrefURL} \url{https://doi.org/10.1093/aristotelian/69.1.149} \end{APACrefURL}
\newblock
\APACrefnote{New Series; paper read at the Aristotelian Society, 28 October 1968}
\PrintBackRefs{\CurrentBib}

\bibitem [\protect \citeauthoryear {%
Lightman%
\ \protect \BOthers {.}}{%
Lightman%
\ \protect \BOthers {.}}{%
{\protect \APACyear {2024}}%
}]{%
lightman2024letsverify}
\APACinsertmetastar {%
lightman2024letsverify}%
\begin{APACrefauthors}%
Lightman, H.%
, Kosaraju, V.%
, Burda, Y.%
, Edwards, H.%
, Baker, B.%
, Lee, T.%
, Leike, J.%
, Schulman, J.%
, Sutskever, I.%
\BCBL {}\ \BBA {} Cobbe, K.%
\end{APACrefauthors}%
\unskip\
\newblock
\APACrefYearMonthDay{2024}{}{}.
\newblock
{\BBOQ}\APACrefatitle {Let's Verify Step by Step} {Let's verify step by step}.{\BBCQ}
\newblock
\BIn{} \APACrefbtitle {International Conference on Learning Representations (ICLR).} {International conference on learning representations (iclr).}
\newblock
\begin{APACrefURL} \url{https://arxiv.org/abs/2305.20050} \end{APACrefURL}
\newblock
\APACrefnote{arXiv:2305.20050}
\PrintBackRefs{\CurrentBib}

\bibitem [\protect \citeauthoryear {%
J.~Liu%
, Xia%
, Wang%
\BCBL {}\ \BBA {} Zhang%
}{%
J.~Liu%
\ \protect \BOthers {.}}{%
{\protect \APACyear {2023}}%
}]{%
liu2023evalplus}
\APACinsertmetastar {%
liu2023evalplus}%
\begin{APACrefauthors}%
Liu, J.%
, Xia, C\BPBI S.%
, Wang, Y.%
\BCBL {}\ \BBA {} Zhang, L.%
\end{APACrefauthors}%
\unskip\
\newblock
\APACrefYearMonthDay{2023}{}{}.
\newblock
{\BBOQ}\APACrefatitle {Is Your Code Generated by {ChatGPT} Really Correct? {R}igorous Evaluation of Large Language Models for Code Generation} {Is your code generated by {ChatGPT} really correct? {R}igorous evaluation of large language models for code generation}.{\BBCQ}
\newblock
\BIn{} \APACrefbtitle {Advances in Neural Information Processing Systems (NeurIPS)} {Advances in neural information processing systems (neurips)}\ (\BVOL~36, \BPGS\ 21558--21572).
\newblock
\begin{APACrefURL} \url{https://proceedings.neurips.cc/paper_files/paper/2023/hash/43e9d647ccd3e4b7b5baab53f0368686-Abstract-Conference.html} \end{APACrefURL}
\newblock
\APACrefnote{arXiv:2305.01210}
\PrintBackRefs{\CurrentBib}

\bibitem [\protect \citeauthoryear {%
X.~Liu%
\ \protect \BOthers {.}}{%
X.~Liu%
\ \protect \BOthers {.}}{%
{\protect \APACyear {2025}}%
}]{%
liu2025testadequacy}
\APACinsertmetastar {%
liu2025testadequacy}%
\begin{APACrefauthors}%
Liu, X.%
, Sun, X.%
, Bo, L.%
, Hu, Y.%
, Liu, X.%
\BCBL {}\ \BBA {} Ye, Z.%
\end{APACrefauthors}%
\unskip\
\newblock
\APACrefYearMonthDay{2025}{}{}.
\newblock
{\BBOQ}\APACrefatitle {Evaluating the Test Adequacy of Benchmarks for {LLMs} on Code Generation} {Evaluating the test adequacy of benchmarks for {LLMs} on code generation}.{\BBCQ}
\newblock
\APACjournalVolNumPages{Journal of Software: Evolution and Process}{37}{7}{e70034}.
\newblock
\begin{APACrefURL} \url{https://doi.org/10.1002/smr.70034} \end{APACrefURL}
\PrintBackRefs{\CurrentBib}

\bibitem [\protect \citeauthoryear {%
Madaan%
\ \protect \BOthers {.}}{%
Madaan%
\ \protect \BOthers {.}}{%
{\protect \APACyear {2023}}%
}]{%
madaan2023selfrefine}
\APACinsertmetastar {%
madaan2023selfrefine}%
\begin{APACrefauthors}%
Madaan, A.%
, Tandon, N.%
, Gupta, P.%
, Hallinan, S.%
, Gao, L.%
, Wiegreffe, S.%
, Alon, U.%
, Dziri, N.%
, Prabhumoye, S.%
, Yang, Y.%
, Gupta, S.%
, Majumder, B\BPBI P.%
, Hermann, K.%
, Welleck, S.%
, Yazdanbakhsh, A.%
\BCBL {}\ \BBA {} Clark, P.%
\end{APACrefauthors}%
\unskip\
\newblock
\APACrefYearMonthDay{2023}{}{}.
\newblock
{\BBOQ}\APACrefatitle {{Self-Refine}: Iterative Refinement with Self-Feedback} {{Self-Refine}: Iterative refinement with self-feedback}.{\BBCQ}
\newblock
\BIn{} \APACrefbtitle {Advances in Neural Information Processing Systems (NeurIPS)} {Advances in neural information processing systems (neurips)}\ (\BVOL~36).
\newblock
\begin{APACrefURL} \url{https://arxiv.org/abs/2303.17651} \end{APACrefURL}
\newblock
\APACrefnote{arXiv:2303.17651}
\PrintBackRefs{\CurrentBib}

\bibitem [\protect \citeauthoryear {%
Matton%
\ \protect \BOthers {.}}{%
Matton%
\ \protect \BOthers {.}}{%
{\protect \APACyear {2024}}%
}]{%
matton2024leakage}
\APACinsertmetastar {%
matton2024leakage}%
\begin{APACrefauthors}%
Matton, A.%
, Sherborne, T.%
, Aumiller, D.%
, Tommasone, E.%
, Alizadeh, M.%
, He, J.%
, Ma, R.%
, Voisin, M.%
, Gilsenan-McMahon, E.%
\BCBL {}\ \BBA {} Gall\'{e}, M.%
\end{APACrefauthors}%
\unskip\
\newblock
\APACrefYearMonthDay{2024}{}{}.
\newblock
{\BBOQ}\APACrefatitle {On Leakage of Code Generation Evaluation Datasets} {On leakage of code generation evaluation datasets}.{\BBCQ}
\newblock
\BIn{} \APACrefbtitle {Findings of the Association for Computational Linguistics: EMNLP 2024} {Findings of the association for computational linguistics: Emnlp 2024}\ (\BPGS\ 13215--13223).
\newblock
\begin{APACrefURL} \url{https://doi.org/10.18653/v1/2024.findings-emnlp.772} \end{APACrefURL}
\newblock
\APACrefnote{arXiv:2407.07565}
\PrintBackRefs{\CurrentBib}

\bibitem [\protect \citeauthoryear {%
Mayo%
}{%
Mayo%
}{%
{\protect \APACyear {2025}}%
}]{%
mayo2025severe}
\APACinsertmetastar {%
mayo2025severe}%
\begin{APACrefauthors}%
Mayo, D\BPBI G.%
\end{APACrefauthors}%
\unskip\
\newblock
\APACrefYearMonthDay{2025}{}{}.
\newblock
{\BBOQ}\APACrefatitle {Severe Testing: Error Statistics versus {Bayes} Factor Tests} {Severe testing: Error statistics versus {Bayes} factor tests}.{\BBCQ}
\newblock
\APACjournalVolNumPages{The British Journal for the Philosophy of Science}{}{}{}.
\newblock
\begin{APACrefURL} \url{https://doi.org/10.1086/736950} \end{APACrefURL}
\newblock
\APACrefnote{Advance online publication; article 736950}
\PrintBackRefs{\CurrentBib}

\bibitem [\protect \citeauthoryear {%
Mayo%
\ \BBA {} Spanos%
}{%
Mayo%
\ \BBA {} Spanos%
}{%
{\protect \APACyear {2006}}%
}]{%
mayospanos2006severe}
\APACinsertmetastar {%
mayospanos2006severe}%
\begin{APACrefauthors}%
Mayo, D\BPBI G.%
\BCBT {}\ \BBA {} Spanos, A.%
\end{APACrefauthors}%
\unskip\
\newblock
\APACrefYearMonthDay{2006}{}{}.
\newblock
{\BBOQ}\APACrefatitle {Severe Testing as a Basic Concept in a {Neyman--Pearson} Philosophy of Induction} {Severe testing as a basic concept in a {Neyman--Pearson} philosophy of induction}.{\BBCQ}
\newblock
\APACjournalVolNumPages{The British Journal for the Philosophy of Science}{57}{2}{323--357}.
\newblock
\begin{APACrefURL} \url{https://doi.org/10.1093/bjps/axl003} \end{APACrefURL}
\PrintBackRefs{\CurrentBib}

\bibitem [\protect \citeauthoryear {%
Niiniluoto%
}{%
Niiniluoto%
}{%
{\protect \APACyear {2014}}%
}]{%
niiniluoto2014progress}
\APACinsertmetastar {%
niiniluoto2014progress}%
\begin{APACrefauthors}%
Niiniluoto, I.%
\end{APACrefauthors}%
\unskip\
\newblock
\APACrefYearMonthDay{2014}{}{}.
\newblock
{\BBOQ}\APACrefatitle {Scientific Progress as Increasing Verisimilitude} {Scientific progress as increasing verisimilitude}.{\BBCQ}
\newblock
\APACjournalVolNumPages{Studies in History and Philosophy of Science Part A}{46}{}{73--77}.
\newblock
\begin{APACrefURL} \url{https://doi.org/10.1016/j.shpsa.2014.02.002} \end{APACrefURL}
\PrintBackRefs{\CurrentBib}

\bibitem [\protect \citeauthoryear {%
Pan%
\ \protect \BOthers {.}}{%
Pan%
\ \protect \BOthers {.}}{%
{\protect \APACyear {2026}}%
}]{%
pan2026coverrl}
\APACinsertmetastar {%
pan2026coverrl}%
\begin{APACrefauthors}%
Pan, T.%
, Yan, Y.%
, Wang, Z.%
, Zhang, R.%
, Hou, G.%
, Zhang, W.%
, Lu, W.%
, Xiao, J.%
\BCBL {}\ \BBA {} Shen, Y.%
\end{APACrefauthors}%
\unskip\
\newblock
\APACrefYearMonthDay{2026}{}{}.
\newblock
\APACrefbtitle {{CoVerRL}: Breaking the Consensus Trap in Label-Free Reasoning via Generator--Verifier Co-Evolution.} {{CoVerRL}: Breaking the consensus trap in label-free reasoning via generator--verifier co-evolution.}
\newblock
\APAChowpublished {arXiv preprint arXiv:2603.17775}.
\newblock
\begin{APACrefURL} \url{https://arxiv.org/abs/2603.17775} \end{APACrefURL}
\PrintBackRefs{\CurrentBib}

\bibitem [\protect \citeauthoryear {%
Popper%
}{%
Popper%
}{%
{\protect \APACyear {1959}}%
}]{%
popper1959logic}
\APACinsertmetastar {%
popper1959logic}%
\begin{APACrefauthors}%
Popper, K\BPBI R.%
\end{APACrefauthors}%
\unskip\
\newblock
\APACrefYear{1959}.
\newblock
\APACrefbtitle {The Logic of Scientific Discovery} {The logic of scientific discovery}.
\newblock
\APACaddressPublisher{London}{Hutchinson}.
\newblock
\APACrefnote{English translation of Logik der Forschung (1934)}
\PrintBackRefs{\CurrentBib}

\bibitem [\protect \citeauthoryear {%
Popper%
}{%
Popper%
}{%
{\protect \APACyear {1963}}%
}]{%
popper1963conjectures}
\APACinsertmetastar {%
popper1963conjectures}%
\begin{APACrefauthors}%
Popper, K\BPBI R.%
\end{APACrefauthors}%
\unskip\
\newblock
\APACrefYear{1963}.
\newblock
\APACrefbtitle {Conjectures and Refutations: The Growth of Scientific Knowledge} {Conjectures and refutations: The growth of scientific knowledge}.
\newblock
\APACaddressPublisher{London}{Routledge and Kegan Paul}.
\PrintBackRefs{\CurrentBib}

\bibitem [\protect \citeauthoryear {%
Sclar%
, Choi%
, Tsvetkov%
\BCBL {}\ \BBA {} Suhr%
}{%
Sclar%
\ \protect \BOthers {.}}{%
{\protect \APACyear {2024}}%
}]{%
sclar2024formspread}
\APACinsertmetastar {%
sclar2024formspread}%
\begin{APACrefauthors}%
Sclar, M.%
, Choi, Y.%
, Tsvetkov, Y.%
\BCBL {}\ \BBA {} Suhr, A.%
\end{APACrefauthors}%
\unskip\
\newblock
\APACrefYearMonthDay{2024}{}{}.
\newblock
{\BBOQ}\APACrefatitle {Quantifying Language Models' Sensitivity to Spurious Features in Prompt Design or: {H}ow {I} Learned to Start Worrying about Prompt Formatting} {Quantifying language models' sensitivity to spurious features in prompt design or: {H}ow {I} learned to start worrying about prompt formatting}.{\BBCQ}
\newblock
\BIn{} \APACrefbtitle {International Conference on Learning Representations (ICLR).} {International conference on learning representations (iclr).}
\newblock
\begin{APACrefURL} \url{https://arxiv.org/abs/2310.11324} \end{APACrefURL}
\newblock
\APACrefnote{arXiv:2310.11324}
\PrintBackRefs{\CurrentBib}

\bibitem [\protect \citeauthoryear {%
Shi%
, Fried%
, Ghazvininejad%
, Zettlemoyer%
\BCBL {}\ \BBA {} Wang%
}{%
Shi%
\ \protect \BOthers {.}}{%
{\protect \APACyear {2022}}%
}]{%
shi2022mbrexec}
\APACinsertmetastar {%
shi2022mbrexec}%
\begin{APACrefauthors}%
Shi, F.%
, Fried, D.%
, Ghazvininejad, M.%
, Zettlemoyer, L.%
\BCBL {}\ \BBA {} Wang, S\BPBI I.%
\end{APACrefauthors}%
\unskip\
\newblock
\APACrefYearMonthDay{2022}{}{}.
\newblock
{\BBOQ}\APACrefatitle {Natural Language to Code Translation with Execution} {Natural language to code translation with execution}.{\BBCQ}
\newblock
\BIn{} \APACrefbtitle {Proceedings of the 2022 Conference on Empirical Methods in Natural Language Processing (EMNLP).} {Proceedings of the 2022 conference on empirical methods in natural language processing (emnlp).}
\newblock
\begin{APACrefURL} \url{https://doi.org/10.18653/v1/2022.emnlp-main.231} \end{APACrefURL}
\newblock
\APACrefnote{MBR-Exec; arXiv:2204.11454}
\PrintBackRefs{\CurrentBib}

\bibitem [\protect \citeauthoryear {%
Shinn%
\ \protect \BOthers {.}}{%
Shinn%
\ \protect \BOthers {.}}{%
{\protect \APACyear {2023}}%
}]{%
shinn2023reflexion}
\APACinsertmetastar {%
shinn2023reflexion}%
\begin{APACrefauthors}%
Shinn, N.%
, Cassano, F.%
, Berman, E.%
, Gopinath, A.%
, Narasimhan, K.%
\BCBL {}\ \BBA {} Yao, S.%
\end{APACrefauthors}%
\unskip\
\newblock
\APACrefYearMonthDay{2023}{}{}.
\newblock
{\BBOQ}\APACrefatitle {{Reflexion}: Language Agents with Verbal Reinforcement Learning} {{Reflexion}: Language agents with verbal reinforcement learning}.{\BBCQ}
\newblock
\BIn{} \APACrefbtitle {Advances in Neural Information Processing Systems (NeurIPS)} {Advances in neural information processing systems (neurips)}\ (\BVOL~36).
\newblock
\begin{APACrefURL} \url{https://proceedings.neurips.cc/paper_files/paper/2023/hash/1b44b878bb782e6954cd888628510e90-Abstract-Conference.html} \end{APACrefURL}
\newblock
\APACrefnote{arXiv:2303.11366}
\PrintBackRefs{\CurrentBib}

\bibitem [\protect \citeauthoryear {%
Sinha%
\ \protect \BOthers {.}}{%
Sinha%
\ \protect \BOthers {.}}{%
{\protect \APACyear {2025}}%
}]{%
sinha2025refute}
\APACinsertmetastar {%
sinha2025refute}%
\begin{APACrefauthors}%
Sinha, S.%
, Goel, S.%
, Kumaraguru, P.%
, Geiping, J.%
, Bethge, M.%
\BCBL {}\ \BBA {} Prabhu, A.%
\end{APACrefauthors}%
\unskip\
\newblock
\APACrefYearMonthDay{2025}{}{}.
\newblock
{\BBOQ}\APACrefatitle {Can Language Models Falsify? Evaluating Algorithmic Reasoning with Counterexample Creation} {Can language models falsify? evaluating algorithmic reasoning with counterexample creation}.{\BBCQ}
\newblock
\BIn{} \APACrefbtitle {Conference on Language Modeling (COLM).} {Conference on language modeling (colm).}
\newblock
\begin{APACrefURL} \url{https://arxiv.org/abs/2502.19414} \end{APACrefURL}
\newblock
\APACrefnote{arXiv:2502.19414}
\PrintBackRefs{\CurrentBib}

\bibitem [\protect \citeauthoryear {%
Snell%
, Lee%
, Xu%
\BCBL {}\ \BBA {} Kumar%
}{%
Snell%
\ \protect \BOthers {.}}{%
{\protect \APACyear {2025}}%
}]{%
snell2024testtime}
\APACinsertmetastar {%
snell2024testtime}%
\begin{APACrefauthors}%
Snell, C.%
, Lee, J.%
, Xu, K.%
\BCBL {}\ \BBA {} Kumar, A.%
\end{APACrefauthors}%
\unskip\
\newblock
\APACrefYearMonthDay{2025}{}{}.
\newblock
{\BBOQ}\APACrefatitle {Scaling {LLM} Test-Time Compute Optimally Can Be More Effective than Scaling Model Parameters} {Scaling {LLM} test-time compute optimally can be more effective than scaling model parameters}.{\BBCQ}
\newblock
\BIn{} \APACrefbtitle {International Conference on Learning Representations (ICLR).} {International conference on learning representations (iclr).}
\newblock
\begin{APACrefURL} \url{https://arxiv.org/abs/2408.03314} \end{APACrefURL}
\newblock
\APACrefnote{arXiv:2408.03314}
\PrintBackRefs{\CurrentBib}

\bibitem [\protect \citeauthoryear {%
Song%
\ \protect \BOthers {.}}{%
Song%
\ \protect \BOthers {.}}{%
{\protect \APACyear {2025}}%
}]{%
song2025mindthegap}
\APACinsertmetastar {%
song2025mindthegap}%
\begin{APACrefauthors}%
Song, Y.%
, Zhang, H.%
, Eisenach, C.%
, Kakade, S.%
, Foster, D.%
\BCBL {}\ \BBA {} Ghai, U.%
\end{APACrefauthors}%
\unskip\
\newblock
\APACrefYearMonthDay{2025}{}{}.
\newblock
{\BBOQ}\APACrefatitle {Mind the Gap: Examining the Self-Improvement Capabilities of Large Language Models} {Mind the gap: Examining the self-improvement capabilities of large language models}.{\BBCQ}
\newblock
\BIn{} \APACrefbtitle {International Conference on Learning Representations (ICLR).} {International conference on learning representations (iclr).}
\newblock
\begin{APACrefURL} \url{https://arxiv.org/abs/2412.02674} \end{APACrefURL}
\newblock
\APACrefnote{arXiv:2412.02674}
\PrintBackRefs{\CurrentBib}

\bibitem [\protect \citeauthoryear {%
Valmeekam%
, Marquez%
\BCBL {}\ \BBA {} Kambhampati%
}{%
Valmeekam%
\ \protect \BOthers {.}}{%
{\protect \APACyear {2023}}%
}]{%
valmeekam2023self}
\APACinsertmetastar {%
valmeekam2023self}%
\begin{APACrefauthors}%
Valmeekam, K.%
, Marquez, M.%
\BCBL {}\ \BBA {} Kambhampati, S.%
\end{APACrefauthors}%
\unskip\
\newblock
\APACrefYearMonthDay{2023}{}{}.
\newblock
\APACrefbtitle {Can Large Language Models Really Improve by Self-Critiquing Their Own Plans?} {Can large language models really improve by self-critiquing their own plans?}
\newblock
\APAChowpublished {arXiv preprint arXiv:2310.08118; NeurIPS 2023 FMDM Workshop}.
\newblock
\begin{APACrefURL} \url{https://arxiv.org/abs/2310.08118} \end{APACrefURL}
\PrintBackRefs{\CurrentBib}

\bibitem [\protect \citeauthoryear {%
Vranje\v{s}%
\ \protect \BOthers {.}}{%
Vranje\v{s}%
\ \protect \BOthers {.}}{%
{\protect \APACyear {2024}}%
}]{%
vranjes2024design}
\APACinsertmetastar {%
vranjes2024design}%
\begin{APACrefauthors}%
Vranje\v{s}, D.%
, Ehrhardt, J.%
, Heesch, R.%
, Moddemann, L.%
, Steude, H\BPBI S.%
\BCBL {}\ \BBA {} Niggemann, O.%
\end{APACrefauthors}%
\unskip\
\newblock
\APACrefYearMonthDay{2024}{}{}.
\newblock
{\BBOQ}\APACrefatitle {Design Principles for Falsifiable, Replicable and Reproducible Empirical Machine Learning Research} {Design principles for falsifiable, replicable and reproducible empirical machine learning research}.{\BBCQ}
\newblock
\BIn{} \APACrefbtitle {35th International Conference on Principles of Diagnosis and Resilient Systems (DX 2024)} {35th international conference on principles of diagnosis and resilient systems (dx 2024)}\ (\BVOL~125, \BPGS\ 7:1--7:13).
\newblock
\APACaddressPublisher{}{Schloss Dagstuhl -- Leibniz-Zentrum f\"{u}r Informatik}.
\newblock
\begin{APACrefURL} \url{https://doi.org/10.4230/OASIcs.DX.2024.7} \end{APACrefURL}
\newblock
\APACrefnote{arXiv:2405.18077}
\PrintBackRefs{\CurrentBib}

\bibitem [\protect \citeauthoryear {%
E.~Wang%
\ \protect \BOthers {.}}{%
E.~Wang%
\ \protect \BOthers {.}}{%
{\protect \APACyear {2025}}%
}]{%
wang2025plansearch}
\APACinsertmetastar {%
wang2025plansearch}%
\begin{APACrefauthors}%
Wang, E.%
, Cassano, F.%
, Wu, C.%
, Bai, Y.%
, Song, W.%
, Nath, V.%
, Han, Z.%
, Hendryx, S.%
, Yue, S.%
\BCBL {}\ \BBA {} Zhang, H.%
\end{APACrefauthors}%
\unskip\
\newblock
\APACrefYearMonthDay{2025}{}{}.
\newblock
{\BBOQ}\APACrefatitle {Planning in Natural Language Improves {LLM} Search for Code Generation} {Planning in natural language improves {LLM} search for code generation}.{\BBCQ}
\newblock
\BIn{} \APACrefbtitle {International Conference on Learning Representations (ICLR).} {International conference on learning representations (iclr).}
\newblock
\begin{APACrefURL} \url{https://arxiv.org/abs/2409.03733} \end{APACrefURL}
\newblock
\APACrefnote{arXiv:2409.03733}
\PrintBackRefs{\CurrentBib}

\bibitem [\protect \citeauthoryear {%
X.~Wang%
\ \protect \BOthers {.}}{%
X.~Wang%
\ \protect \BOthers {.}}{%
{\protect \APACyear {2023}}%
}]{%
wang2022selfconsistency}
\APACinsertmetastar {%
wang2022selfconsistency}%
\begin{APACrefauthors}%
Wang, X.%
, Wei, J.%
, Schuurmans, D.%
, Le, Q\BPBI V.%
, Chi, E\BPBI H.%
, Narang, S.%
, Chowdhery, A.%
\BCBL {}\ \BBA {} Zhou, D.%
\end{APACrefauthors}%
\unskip\
\newblock
\APACrefYearMonthDay{2023}{}{}.
\newblock
{\BBOQ}\APACrefatitle {Self-Consistency Improves Chain of Thought Reasoning in Language Models} {Self-consistency improves chain of thought reasoning in language models}.{\BBCQ}
\newblock
\BIn{} \APACrefbtitle {International Conference on Learning Representations (ICLR).} {International conference on learning representations (iclr).}
\newblock
\begin{APACrefURL} \url{https://arxiv.org/abs/2203.11171} \end{APACrefURL}
\newblock
\APACrefnote{arXiv:2203.11171}
\PrintBackRefs{\CurrentBib}

\bibitem [\protect \citeauthoryear {%
Q.~Zhang%
\ \protect \BOthers {.}}{%
Q.~Zhang%
\ \protect \BOthers {.}}{%
{\protect \APACyear {2025}}%
}]{%
zhang2025darkside}
\APACinsertmetastar {%
zhang2025darkside}%
\begin{APACrefauthors}%
Zhang, Q.%
, Wang, D.%
, Qian, H.%
, Li, Y.%
, Zhang, T.%
, Huang, M.%
, Xu, K.%
, Li, H.%
, Liu, Y.%
\BCBL {}\ \BBA {} Qiu, H.%
\end{APACrefauthors}%
\unskip\
\newblock
\APACrefYearMonthDay{2025}{}{}.
\newblock
{\BBOQ}\APACrefatitle {Understanding the Dark Side of {LLMs}' Intrinsic Self-Correction} {Understanding the dark side of {LLMs}' intrinsic self-correction}.{\BBCQ}
\newblock
\BIn{} \APACrefbtitle {Proceedings of the 63rd Annual Meeting of the Association for Computational Linguistics (Volume 1: Long Papers)} {Proceedings of the 63rd annual meeting of the association for computational linguistics (volume 1: Long papers)}\ (\BPGS\ 27066--27101).
\newblock
\begin{APACrefURL} \url{https://doi.org/10.18653/v1/2025.acl-long.1314} \end{APACrefURL}
\newblock
\APACrefnote{arXiv:2412.14959}
\PrintBackRefs{\CurrentBib}

\bibitem [\protect \citeauthoryear {%
T.~Zhang%
\ \protect \BOthers {.}}{%
T.~Zhang%
\ \protect \BOthers {.}}{%
{\protect \APACyear {2023}}%
}]{%
zhang2023coderreviewer}
\APACinsertmetastar {%
zhang2023coderreviewer}%
\begin{APACrefauthors}%
Zhang, T.%
, Yu, T.%
, Hashimoto, T\BPBI B.%
, Lewis, M.%
, tau Yih, W.%
, Fried, D.%
\BCBL {}\ \BBA {} Wang, S\BPBI I.%
\end{APACrefauthors}%
\unskip\
\newblock
\APACrefYearMonthDay{2023}{}{}.
\newblock
{\BBOQ}\APACrefatitle {Coder Reviewer Reranking for Code Generation} {Coder reviewer reranking for code generation}.{\BBCQ}
\newblock
\BIn{} \APACrefbtitle {Proceedings of the 40th International Conference on Machine Learning (ICML)} {Proceedings of the 40th international conference on machine learning (icml)}\ (\BVOL~202, \BPGS\ 41832--41846).
\newblock
\begin{APACrefURL} \url{https://proceedings.mlr.press/v202/zhang23av.html} \end{APACrefURL}
\newblock
\APACrefnote{arXiv:2211.16490}
\PrintBackRefs{\CurrentBib}

\end{thebibliography}

\end{document}